\documentclass[sigconf,twocolumn, nonacm]{acmart}

\usepackage{amsthm}
\usepackage{amsmath}
\usepackage{enumerate}
\usepackage{tikz}
\usetikzlibrary{pgfplots.groupplots}
\usetikzlibrary{decorations.markings}
\usetikzlibrary{patterns}
\usepackage[edges]{forest}
\usepackage{subcaption}
\usepackage[ruled, vlined, linesnumbered, algo2e]{algorithm2e}
\usepackage{pgfplotstable}
\usepackage{algpseudocode}
\usepackage{xfrac}
\usepackage{multicol}
\usepackage{pifont}
\usepackage{xcolor,colortbl}
\usepackage{mdframed}
\usepackage[most]{tcolorbox}

\definecolor{Gray}{gray}{0.98}

\newcommand{\cmark}{\ding{51}}%
\newcommand{\xmark}{\ding{55}}%
\newcommand{\rom}[1]{\lowercase\expandafter{\romannumeral #1\relax}}

\usetikzlibrary{decorations.markings}
\tikzset{
  set arrow inside/.code={\pgfqkeys{/tikz/arrow inside}{#1}},
  set arrow inside={end/.initial=>, opt/.initial=},
  /pgf/decoration/Mark/.style={
      mark/.expanded=at position #1 with
      {
          \noexpand\arrow[\pgfkeysvalueof{/tikz/arrow inside/opt}]{\pgfkeysvalueof{/tikz/arrow inside/end}}
      }
  },
  arrow inside/.style 2 args={
      set arrow inside={#1},
      postaction={
          decorate,decoration={
              markings,Mark/.list={#2}
          }
      }
  },
}

\tikzset{
roundnode/.style={circle, draw=green!60, fill=green!5, 
very thick, minimum size=8mm},
font={\fontsize{15pt}{12}\selectfont},
fontscale/.style = {font=\relsize{#1}}
}

\AtBeginDocument{%
  }

\begin{document}

\setcounter{figure}{0}   

\author{Qing Chen}
\affiliation{%
  \institution{University of Zurich}
  \country{Switzerland}
}
\email{qing@ifi.uzh.ch}

\author{Michael H. B{\"o}hlen}
\affiliation{%
  \institution{University of Zurich}
  \country{Switzerland}
}
\email{boehlen@ifi.uzh.ch}

\author{Sven Helmer}
\affiliation{%
  \institution{University of Zurich}
  \country{Switzerland}
}
\email{helmer@ifi.uzh.ch}


\title{An experimental comparison of tree-data structures for
  connectivity queries on fully-dynamic undirected graphs (Extended Version)}

\begin{abstract}
  During the past decades significant efforts have been made to
  propose data structures for answering connectivity queries on fully
  dynamic graphs, i.e., graphs with frequent insertions and deletions
  of edges.  However, a comprehensive understanding of how these data
  structures perform in practice is missing, since not all of them
  have been implemented, let alone evaluated experimentally.  We
  provide reference implementations for the proposed data structures
  and experimentally evaluate them on a wide range of graphs.  Our
  findings show that the current solutions are not ready to be
  deployed in systems as is, as every data structure has critical
  weaknesses when used in practice.  Key limitations that must be
  overcome are the space and time overhead incurred by balanced data
  structures, the degeneration of the runtime of space-efficient data
  structures in worst case scenarios, and the maintenance costs for
  balanced data structures.  We detail our findings in the
  experimental evaluation and provide recommendations for implementing
  robust solutions for answering connectivity queries on dynamic
  graphs.
\end{abstract}

\settopmatter{printfolios=true} 


\pagenumbering{gobble}

\maketitle
\pagenumbering{arabic}
\setcounter{page}{1}

The source code, data, and/or other artifacts have been made available at
\url{https://github.com/qingchen3/Imp_bench_dyn}.

\section{Introduction}
\label{sec:intro}

It is common lore that asymptotically faster algorithms are preferable
over asymptotically slower alternatives.  Since many years this
understanding has served our community as a coarse guideline for
designing and selecting data structures and algorithms.  In recent
years, however, a number of solutions have been proposed where more
fine-grained evaluations are needed to gain actionable insights into
the properties of data structures and algorithms~\cite{Rough19,
  augsten2010tasm, xiao19PODS, Ding21SIGMOD, fuchs2022sortledton}.

A problem where insights beyond worst case asymptotic complexity and
amortized costs are missing is the connectivity problem for fully
dynamic graphs.  Connectivity queries over fully dynamic graphs
consider graphs with possibly frequent edge insertions and deletions
and determine if two vertices are connected. While this fundamental
problem has been studied extensively on a theoretical level for
dynamic graphs \cite{Fan22,Patra04, eppstein1993separator,
  henzinger1998lower, patrascu2007planning,
  eppstein1992sparsification} and various data structures
\cite{henzinger1995randomized, wulff2013faster, thorup2000near,
  huang2017fully, dtree} and algorithms \cite{HDT, huang2023fully,
  henzinger1997sampling} have been proposed, we lack an understanding
of how these solutions perform in practice.  Typical
  application areas for fully dynamic graphs are \emph{simulation
    scenarios}.  For instance, power grids are networks with large
  diameters that include microgrids with smaller diameters.  To test
  the robustness of power grids, islanding techniques \cite{Goderya80,
    Mahat08} are being used to detect the intentional or unintentional
  division of a connected power grid into disconnected regions.  This
  is done by simulating the deletion of one or multiple edges and
  running connectivity queries for each possible set of edge
  modifications to determine if the power network remains connected.
  Such simulations crucially depend on fast connectivity algorithms
  since many possible sets of edge deletions must be performed and
  checked by running a large number of connectivity queries.

To solve the dynamic connectivity problem three classes of solutions,
all based on spanning trees, have been proposed (see
Section~\ref{sec:maintaining_ds} for details).  The first class
maintains spanning trees without levels.  D-trees
  \cite{dtree} use a \emph{heuristic that attempts to minimize the
  sum of distances} between the root and all other nodes in spanning
trees,  whereas link-cut trees \cite{Sleator_Tarjan_85} decompose
  spanning trees into a set of paths and nodes.  The second class is
based on \emph{Euler tour trees} \cite{henzinger1995randomized}, which
are balanced binary trees together with a partitioning of the edges
into levels, to get the first polylogarithmic time bound for deleting
edges \cite{henzinger1999randomized}.  The third category are
\emph{height-bounded trees} to improve the time bounds of balanced
binary trees.  Height-bounded trees come in two variants: structural
trees \cite{thorup2000near, wulff2013faster} are $k$-ary trees ($k$
not being constrained) while local trees \cite{thorup2000near,
  huang2017fully, wulff2013faster, huang2023fully} are binary trees
based on rank trees.

In this empirical study we comprehensively evaluate the performance of
solutions for the dynamic connectivity problem.  Our goal is to
understand the tradeoffs along three main dimensions.  The first
dimension is the balancedness of spanning trees.  Balanced trees
provide guarantees for the performance of operations but do not allow
to directly represent the edges of the graph by edges of the spanning
tree since the organization of the nodes in the spanning tree is
determined by the balancedness criteria.  This leads to a
non-negligible overhead for balanced trees.  The second dimension is a
partitioned representation of spanning trees.  A partitioning of the
nodes and edges into multiple levels yields faster operations per
level but additional data structures are needed to support the
partitioning and the construction and maintenance of the partitioning
when the graph changes are expensive.  The third dimension are space-
versus deletion-efficient data structures and algorithms.  The most
time-consuming spanning tree operation is the deletion of an edge that
breaks a component into two components.  In this case it must be
checked if a replacement edge exists that reconnects the two
components.  Deletion-efficient data structures focus on solutions to
improve the worst case performance of such deletions by introducing
auxiliary data structures and considering amortized costs.  While the
approaches successfully curb the worst case performance and the
amortized costs for some classes of workloads they come with hefty
overheads. Space-efficient data structures (D-tree, LCT, ST)
  trade worst case guarantees for lightweight data structures and
  algorithms.  We find that so far 
  no solution achieves a good trade-off among
  memory footprints, tree heights and maintenance cost.

We provide the first comprehensive experimental study for all major
data structures for connectivity queries to guide future research on
this topic.  Our main technical contributions can be summarized as
follows:
\begin{itemize}
\item We implement all major data structures for the dynamic
  connectivity problem: D-tree, link-cut tree LCT, HK, HKS
  (a simplified version of HK), HDT, structural tree ST, a variant STV
  of structural trees, local tree LT, a variant LTV of local trees,
  and lazy local trees LzT.  We provide reference implementations for
  all solutions, some of which have never been implemented before.
\item We extensively evaluate all major data structures on large
  real-world and synthetic graphs with a wide range of workloads. We
  generate workloads that decouple the dependency between the
  insertion and deletion of edges to permit a fine-grained
  control of the growth rate of the graph, and show that lazy local
    trees with the lowest amortized costs are the slowest in terms of
    empirically determined runtime.
\item We leverage our insights from extensive implementations and
  evaluations to offer lessons learned and we 
  provide recommendations for
  future work to pave the path for the first practical 
  and robust data structure for connectivity queries over 
  fully dynamic graphs.
\end{itemize}

\section{Background}
\label{sec:preliminary}

We consider undirected unweighted simple graphs $G = (V, E)$ defined by a
set $V$ of vertices and a set $E$ of edges 
\cite{gibbons1985algorithmic, west2001introduction, 
bonifati2022querying}.  In an
\emph{undirected simple graph}, $(u, v)$ and $(v, u)$ are the 
same edge.  $adj(u)$ denotes the set of vertices
that are directly connected to vertex $u$, $i.e.,$ $adj(u)$ $=$
$\{v~|~v \in V, (u, v)\in E\}$. A \emph{path} $P$ is a sequence of $m$
distinct vertices ($v_1$, $v_2$, ..., $v_m$) where $v_i$ $\in$ $V$ and
every two neighboring vertices $v_i$ and $v_{i+1}$ are connected by
edges $(v_i, v_{i+1})$ $\in$ $E$.  If there are
edges connecting $v_1$ and $v_m$, the sequence ($v_1$, $v_2$, ...,
$v_m$, $v_1$) is a \emph{cycle}.  The \emph{diameter} of a graph is
the length of the longest shortest path between two vertices.  A
\emph{connected component} $C = (V_c, E_c)$ of a graph $G = (V, E)$ is a
maximal subgraph, with $V_c \subseteq V$ and $E_c \subseteq E$, such
that any pair of vertices in $C$ is connected by a path.  A
\emph{tree} is an undirected graph in which any pair of nodes is
connected by exactly one path and there are no cycles.  In a
\emph{rooted tree} there is a designated root node.  Given a connected
component $C = (V_c, E_c)$, a \emph{spanning tree} for $C$ is a rooted
tree $st = (V', E')$ with $V' = V_c$ and $E'$ $\subseteq$ $E_c$.  
An edge $(u, v)$ is a \emph{tree edge} if $(u, v)$
$\in$ $E'$, otherwise it is a \emph{non-tree edge}.  
If the insertion of an edge $(u, v)$ connects two spanning trees, 
$(u, v)$ is a tree edge for the merged spanning tree.
For example, in 
Figure~\ref{fig:exam_c_t}, edges (1, 2), (1, 3) and (3, 6) are
non-tree edges while all other edges are tree edges.  
A \emph{spanning forest} is a set of spanning trees. 

\setcounter{figure}{0}  
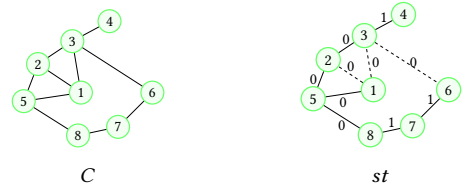
\begin{figure}[htb!] 
  \begin{subfigure}[b]{0.45\columnwidth}\centering
  \scalebox{0.38}{  
  \begin{tikzpicture}
    \node[roundnode] (1)  at (-1.5, -1.0)  {1};
    \node[roundnode] (2)  at (-3, 0.0)   {2};
    \node[roundnode] (3)  at (-1.8, 0.8)  {3};
    \node[roundnode] (4)  at (-0.5, 1.5) {4};
    \node[roundnode] (5)  at (-3.5, -1.3)  {5};

    \node[roundnode] (6)  at (1.0, -1.0) {6};
    \node[roundnode] (7)  at (-0.2, -2.2)  {7};
    \node[roundnode] (8)  at (-1.6, -2.5)  {8};

    \path[-] (1) edge node[above] {} (2);
    \path[-] (1) edge node[midway, right] {} (3);
    \path[-] (2) edge node[above] {} (3);
    \path[-] (3) edge node[above] {} (4);
    \path[-] (2) edge node[left] {} (5);
    \path[-] (1) edge node[below] {} (5);
    \path[-] (3) edge node[right] {} (6);
    
    \path[-] (7) edge node[above] {} (6);
    \path[-] (7) edge node[above] {} (8);
    \path[-] (8) edge node[below] {} (5);
  \end{tikzpicture}
  }
  \caption*{$C$}
  \end{subfigure}
  \begin{subfigure}[b]{0.45\columnwidth}\centering
    \scalebox{0.40}{  
    \begin{tikzpicture}
      \node[roundnode] (1)  at (-1.5, -1.0)  {1};
      \node[roundnode] (2)  at (-3, 0.0)  {2};
      \node[roundnode] (3)  at (-1.8, 0.8)  {3};
      \node[roundnode] (4)  at (-0.5, 1.5)  {4};
      \node[roundnode] (5)  at (-3.5, -1.3)  {5};

      \path[-] (1) edge node[below] {0} (5);
      \path[-] (2) edge node[above] {0} (3);
      \path[-] (3) edge node[above] {1} (4);
      \path[-] (2) edge node[left] {0} (5);
      
    
      \node[roundnode] (6)  at (1.0, -1.0) {6};
      \node[roundnode] (7)  at (-0.2, -2.2) {7};
      \node[roundnode] (8)  at (-1.6, -2.5) {8};
      
      \path[-] (7) edge node[above] {1} (6);
      \path[-] (7) edge node[above] {1} (8);
      \path[-] (8) edge node[below] {0} (5);

      \path[dashed] (3) edge node[right] {0} (6);
      \path[dashed] (1) edge node[above] {0} (2);
      \path[dashed] (1) edge node[midway, right] {0} (3);
    \end{tikzpicture}
    }
    \caption*{$st$}
  \end{subfigure}
  
  \caption{Connected component $C$ and spanning tree $st$ for $C$. 
  Edges are labeled with their levels (described below). Dashed edges
  are non-tree edges for $st$.}
  \label{fig:exam_c_t}
\end{figure}

\subsection{Connectivity queries on tree-data structures}

Given a graph, a connectivity query for two vertices returns true if
there exists a path between the two vertices, otherwise false.  All
tree-data structures maintain a rooted tree for each connected
component of the graph. Answering connectivity queries using tree-data
structures boils down to checking if two nodes have the same root
node, $i.e.,$ they are in the same tree.

\begin{definition}[Connectivity queries on tree-data structures]
  Given a tree-data structure for a graph $G=(V, E)$ and two nodes
  $u$, $v$ $\in$ $V$, the connectivity query $conn(u, v)$ returns True
  if nodes $u$ and $v$ have the same root node in the tree, and False
  otherwise.
\end{definition}

Answering connectivity queries on tree-data structures is traversing
to the root node and hence the query performance is bounded by the
tree heights.  The data structures are based on spanning
trees and need to maintain themselves 
due to the insertions and deletions of edges. Inserting and
deleting a non-tree edge does not change the spanning tree and 
hence is trivial to handle. The insertion of a tree edge 
merges two spanning trees. Deleting a tree edge splits
a spanning tree into two, a replacement edge that reconnects 
the two spanning trees is searched.

\subsection{Partitioning Edges by Levels}
\label{sec:partitioning_level}

Existing works maintain data structures and assign a non-negative 
integer, called $level$, to every
edge. Placing edges at different levels yields better theoretical
amortized costs \cite{tarjan1985amortized} 
when searching for a replacement edge after a tree
edge has been deleted (details follow in
Section~\ref{sec:delete_edge}).  The level of an edge $(u, v)$,
denoted as $l(u, v)$, is 0 when $(u, v)$ is inserted.  When an edge
$(x, y)$ is deleted, we do an exhaustive search, breadth first search
(BFS) or depth first search (DFS), from $x$ or $y$ to find a
replacement edge.  During this search, if $(u, v)$ is not traversed,
$l(u, v)$ remains unchanged.  If $(u, v)$ is traversed, $l(u, v)$ is
increased by 1 up to a bound that depends on the heuristics used by
the method~\cite{henzinger1999randomized, wulff2013faster,
  thorup2000near, huang2017fully}.  Let $l_{max}$ be the maximum level
on edges.  Increasing the level of an edge can be interpreted as
pushing down the edge.  Edges with a larger level are considered
before edges with a smaller level.  Levels of edges vary since some
edges are traversed multiple times while other edges are never
traversed during the search for replacement edges. 
We write $E'_i$ to refer to level-$i$ edges in $E'$, hence
$E'_i$ $\cap$ $E'_j$ $=$ $\emptyset$ for 0 $\leq$ $i$ $<$ $j$ $\leq$
$l_{max}$, and $E'$ = $\cup_{i=0}^{i \leq l_{max}}$ $E'_i$.  The level
of a node $u$, denoted as $l(u)$, is equal to 1 plus the maximum level
on edges that are directly connected to $u$, $i.e.,$ $l(u)$ $=$ 1 $+$
$max\{l(u, x)$ | $(u, x)$ $\in$ $E'$$\}$.  We use the term vertices
for graphs and the term nodes for trees.

We define level-$i$ tree and non-tree neighbors for each node to
facilitate the traversal over level-$i$ nodes.  Without
differentiating between level-$i$ tree and non-tree neighbors, all
level-$i$ edges must be searched during traversals.

\begin{definition} [level-$i$ tree and non-tree neighbors]
  Given a spanning tree $st = (V', E')$ with $E'$ = 
  $\cup_{i=0}^{i \leq l_{max}}$ $E'_i$, the \emph{level-$i$ tree
    neighbors} of a node $u$ are $adj^t_i(u)$ $=$
  $\{ v~|~v \in V', (u, v) \in E'_i \}$.  The \emph{level-$i$ non-tree
    neighbors} of a node $u$ are
  $adj^{nt}_i(u) = \{ v~|~v \in V', l(u, v) = i, (u, v) \not\in E'_i
  \}$.
\end{definition}

\begin{definition} [level-$i$ neighbors]
  The \emph{level-$i$ neighbors} of a node $u$ are
  $adj_i(u) = adj^t_i(u) \cup adj^{nt}_i(u)$.
\end{definition}

\begin{example}
  In Figure~\ref{fig:exam_c_t}, for vertex 3 $adj^t_0(3)$ = $\{2\}$,
  $adj^{nt}_0(3)$ = $\{1, 6\}$, $adj^t_1(3)$ = $\{4\}$, and
  $adj^{nt}_1(3)$ is empty.  The level-0 tree edge $(2, 3)$ is stored
  in $adj^t_0(2)$ and $adj^t_0(3)$, respectively, as 3 $\in$
  $adj^t_0(2)$ and 2 $\in$ $adj^t_0(3)$.  Similarly, the
  level-0 non-tree edge $(1, 3)$ is stored in $adj^{nt}_0(1)$ and
  $adj^{nt}_0(3)$, respectively.
\end{example}

Existing data structures maintain spanning trees, cumulative 
spanning trees, or recursive spanning trees.  These trees keep 
track of the nodes in spanning trees (for details, see
Section~\ref{sec:maintaining_ds}).  Edges with levels are maintained
in two different ways.  One way maintains the edges cumulatively,
i.e., edges with larger levels and those with 
smaller levels are maintained together. 
The other way maintains the edges at each level separately.

\subsection{Cumulative Spanning Trees}

\begin{definition} [Level-i Cumulative Spanning Tree]
  Consider a spanning tree $st = (V', E')$ with nodes $V'$ and edges
  $E'$ = $\cup_{i=0}^{i \leq l_{max}}$ $E'_i$.  The level-$i$
  cumulative spanning tree for $st$ is $_c\mathcal{ST}_i$ $=$
  $(_c\mathcal{V}_i,~ _c\mathcal{E}_i)$ with $_c\mathcal{E}_i$ =
  $\cup_{j=i}^{j \leq l_{max}}$ $E'_j$ and $_c\mathcal{V}_i$ $=$
  $\cup_{j=i}^{j \leq l_{max}}$
  $\{v~|~v \in V', \exists x \in V', (v, x)\in E'_j \}$.
\end{definition}

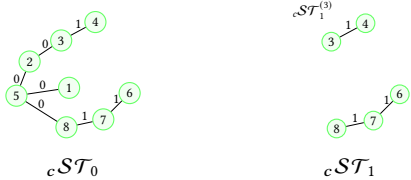
\begin{figure}[htb!] 
  \begin{subfigure}[b]{0.45\columnwidth}\centering
  \scalebox{0.35}{  
  \begin{tikzpicture}
    \node[roundnode] (1)  at (-1.5, -1.0)  {1};
    \node[roundnode] (2)  at (-3, 0.0)  {2};
    \node[roundnode] (3)  at (-1.8, 0.8)  {3};
    \node[roundnode] (4)  at (-0.5, 1.5)  {4};
    \node[roundnode] (5)  at (-3.5, -1.3)  {5};

    \path[-] (1) edge node[above] {0} (5);
    \path[-] (2) edge node[above] {0} (3);
    \path[-] (3) edge node[above] {1} (4);
    \path[-] (2) edge node[left] {0} (5);
  
    \node[roundnode] (6)  at (0.8, -1.2) {6};
    \node[roundnode] (7)  at (-0.2, -2.2) {7};
    \node[roundnode] (8)  at (-1.6, -2.5) {8};
    
    \path[-] (7) edge node[above] {1} (6);
    \path[-] (7) edge node[above] {1} (8);
    \path[-] (8) edge node[above] {0} (5);
  \end{tikzpicture}
  }
  \caption*{$_c\mathcal{ST}_0$}
  \end{subfigure}
  \begin{subfigure}[b]{0.40\columnwidth}\centering
    \scalebox{0.35}{ 
    \begin{tikzpicture}[roundnode/.style={circle, draw=green!60, 
      fill=green!5, very thick, minimum size=4mm}]
      \node[roundnode] (3)  at (-1.8, 0.8)  {3};
      \node[roundnode] (4)  at (-0.5, 1.5)  {4};
      \path[-] (3) edge node[above] {1} (4);
      \node () at (-2.5, 2.0) {$_c\mathcal{ST}_1^{(3)}$};

      \node[roundnode] (6)  at (0.8, -1.2) {6};
      \node[roundnode] (7)  at (-0.2, -2.2) {7};
      \node[roundnode] (8)  at (-1.6, -2.5) {8};
      
      \path[-] (7) edge node[above] {1} (6);
      \path[-] (7) edge node[above] {1} (8);
    \end{tikzpicture}
    }
    \caption*{$_c\mathcal{ST}_1$}
  \end{subfigure}
  
  \caption{Cumulative spanning trees for 
  $C$ in Figure~\ref{fig:exam_c_t}. $_c\mathcal{ST}_1^{(3)}$ denotes the
  level-1 cumulative spanning tree that contains node 3.}
  \label{fig:exam_cst}
\end{figure}

\begin{example}
  Figure~\ref{fig:exam_cst} shows the level-0 cumulative spanning tree
  $_c\mathcal{ST}^0$ for $C$ in Figure~\ref{fig:exam_c_t} with level-0
  and level-1 edges.  In $_c\mathcal{ST}^0$, vertex 3 is a level-2
  vertex since edge $(3, 4)$ has the maximal level, and vertex 2 is a
  level-1 vertex since levels of all edges directly connected to
  vertex 2 are 0.
\end{example}

\subsection{Recursive Spanning Trees}

Recursive spanning trees use super nodes to maintain the edges and
nodes of the next larger level.  The level-$l_{max}$ recursive
spanning tree contains level-$l_{max}$ edges and level-($l_{max} + 1$)
nodes.  All level-$l_{max}$ recursive spanning trees are super nodes
of level-$l_{max}$ in level-($l_{max} - 1$) spanning trees.  Thus, a
level-($l_{max} - 1$) spanning tree contains level-($l_{max} - 1$)
edges, level-$l_{max}$ nodes, and level-$l_{max}$ super nodes.

\begin{definition}[Level-i Recursive Spanning Tree]
  Given a spanning tree $st = (V', E')$ with $E'$ =
   $\cup_{i=0}^{i \leq l_{max}}$ $E'_i$, 
  a level-$i$ recursive spanning
  tree for $st$ is $_r\mathcal{ST}_i$ $=$ $(_r\mathcal{V}_i,~ E_i')$
  where for each $v$ $\in$ $_r\mathcal{V}^i$, (1) $v$ is a
  level-$(i+1)$super node (cf.\ Definition~\ref{def:1}) or (2) $v$
  $\in$ $V'$ and $v$ is a level-$(i+1)$ node.  Any pair of nodes in
  $_r\mathcal{V}_{i}$ are connected via level-$i$ edges.
\end{definition}

Figure~\ref{fig:exam_rst} shows a level-$0$ spanning tree
$_r\mathcal{ST}_0$ in which $s_1$ and $s_1'$ are level-$1$ super nodes
while nodes 1, 2, and 5 are level-$(i+1)$ nodes.  All nodes of
$_r\mathcal{ST}_0$ are connected via level-$0$ edges.

\begin{figure}[htb!]
  \begin{subfigure}{0.30\columnwidth}\centering
    \scalebox{0.35}{
    \begin{tikzpicture}[roundnode/.style={circle, draw=green!60, 
    fill=green!5, very thick, minimum size=4mm}]
      \node[roundnode] (1)  at (-1.5, -1.0)  {1};
      \node[roundnode] (2)  at (-3, 0.0)  {2};
      \node[roundnode] (5)  at (-3.5, -1.3)  {5};
      \node[roundnode] (s1)  at (-1.6, 0.9)  {$s_1$};
      \path[-] (1) edge node[above] {0} (5);
      \path[-] (2) edge node[above] {0} (s1) ;
      \path[-] (2) edge node[left] {0} (5);
      \node[roundnode] (s1p)  at (-1.6, -2.5) {$s_1'$};
      \path[-] (s1p) edge node[below] {0} (5);
    \end{tikzpicture}
    }
    \caption*{$_r\mathcal{ST}_0$}
  \end{subfigure}
  \begin{subfigure}[b]{0.65\columnwidth}\centering
    \scalebox{0.35}{ 
    \begin{tikzpicture}[roundnode/.style={circle, draw=green!60, 
      fill=green!5, very thick, minimum size=4mm}]
      \node[roundnode] (3)  at (-3.8, 0.3)  {3};
      \node[roundnode] (4)  at (-2.5, 1.0)  {4};
      \path[-] (3) edge node[above] {1} (4);
      \draw [rotate=25] (-2.5, 1.85) ellipse (1.4cm and 0.7cm);
      \node  at (-3.0, -0.5)  {$s_1$};
      \node () at (-2.5, 2.0) {$_r\mathcal{ST}_1^{(3)}$};

      \node[roundnode] (6)  at (1.8, 0.8) {6};
      \node[roundnode] (7)  at (0.8, -0.2) {7};
      \node[roundnode] (8)  at (-0.6, -0.5) {8};
      \draw [rotate=27] (0.6, -0.25) ellipse (2.1cm and 0.8cm);
      \node  at  (1.5, -0.7)  {$s_1'$};
      
      \path[-] (7) edge node[above] {1} (6);
      \path[-] (7) edge node[above] {1} (8);
    \end{tikzpicture}
    }
    \caption*{$_r\mathcal{ST}_1$}
  \end{subfigure}
  \caption{Recursive spanning trees for $C$ in 
  Figure~\ref{fig:exam_c_t}. $_r\mathcal{ST}_1^{(3)}$ denotes the
  level-1 recursive spanning tree that contains node 3.}
  \label{fig:exam_rst}
\end{figure}
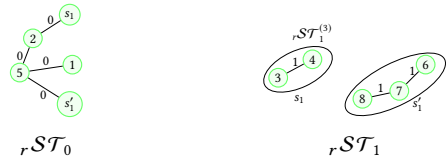

\begin{definition}[Level-(i+1) super node]\label{def:1}
  Given a spanning tree $st = (V', E')$ with $E'$ =
  $\cup_{i=0}^{i \leq l_{max}}$ $E'_i$, 
  a level-$(i+1)$ super node is
  $s_i$ $=$ $(\mathcal{V}_{s_i}, E_{s_i}')$ with $E_{s_i}'$
  $\subseteq$ $E_{i+1}'$.  For each node $v$ $\in$
  $\mathcal{V}_{s_i}$, $v$ is a level-$(i+2)$ super node or a
  level-$(i+2)$ node in $V'$.  Any pair of nodes in
  $\mathcal{V}_{s_i}$ are connected via level-($i+1$) edges.
\end{definition}

To simplify notation, we write $st_i$ to denote either a level-$i$
cumulative or a recursive spanning tree. If there are more than
one level-$i$ spanning trees, $st_i$ is a spanning forest. We use
$st_i^{(u)}$ to denote the level-$i$ spanning tree that contains
node $u$.

\subsection{Inserting Edges}

When an edge $(u, v)$ is inserted, we must determine if
$(u, v)$ is a tree or non-tree edge when it is added to the spanning
tree.  If it is a tree edge the structure of the spanning tree changes
since components must be merged.  To check if $(u, v)$ is a non-tree
edge we run a connectivity query $conn(u, v)$ on level-0 spanning
trees.  If $conn(u, v)$ returns true, $(u, v)$ is a non-tree edge
otherwise $(u, v)$ is a tree edge.  The levels of newly inserted edges
are set to 0.  If $(u, v)$ is a non-tree edge, we add $u$ to
$adj^{nt}_0(v)$ and $v$ to $adj^{nt}_0(u)$.  If $(u, v)$ is a tree
edge, inserting $(u, v)$ merges the level-0 spanning trees that
contain, respectively, $u$ and $v$. For some data structures, edges
with small levels are pushed down after the insertions of edges.

\subsection{Deleting Edges}
\label{sec:delete_edge}
When deleting a level-$i$ edge $(u, v)$, we first determine 
if $(u, v)$ is a non-tree edge or a tree edge. 
Edge $(u, v)$ is a level-$i$ non-tree edge if $v$ is in 
$adj^{nt}_i(u)$ and is a level-$i$ tree edge otherwise.  If $(u, v)$ 
is a level-$i$ non-tree edge, deleting the edge does not change the 
spanning tree.
We only need to remove $u$ and $v$ from $adj^{nt}_i(v)$ and
$adj^{nt}_i(u)$, respectively.  If $(u, v)$ is a level-$i$ tree edge,
the level-$i$ spanning tree $st_i^{(u, v)}$ containing $u$ and $v$ is
split into two level-$i$ spanning trees: $st_i^{(u)}$ containing $u$
and $st_i^{(v)}$ containing $v$.  Let $st_i^{(u)}$ be the tree with
fewer nodes.  We need to know if there is a level-$i$ non-tree edge,
called \emph{replacement edge}, that reconnects $st_i^{(u)}$ and
$st_i^{(v)}$.  We search $st_i^{(u)}$ for a replacement edge. If there
is a level-$i$ non-tree edge $(tx, ty)$ that reconnects $st_i^{(u)}$
and $st_i^{(v)}$, we stop the search and turn $(tx, ty)$ into a
level-$i$ tree edge, otherwise we move to level $i-1$ to search for a
level-$(i-1)$ non-tree edge that reconnects $st_i^{(u)}$ and
$st_i^{(v)}$ until we either reach level 0 or find a replacement edge.
In both situations, the level-$i$ tree edges and level-$i$ non-tree
edges traversed during the search are pushed to level-$(i+1)$.
Pushing level-$i$ edges to level-$(i + 1)$ saves costs for future
deletions of level-$i$ edges since we do not search edges whose levels
are larger than $i$.  Algorithm~\ref{alg:delete} shows the general
procedure for deleting a level-$i$ tree edge.

\begin{algorithm2e}[htb!]
  \small
  \caption{Delete level-$i$ tree edge $(u, v, i)$}
  \label{alg:delete}
  Remove $u$ and $v$ from $adj^t_i(v)$ and $adj^t_i(u)$, respectively\;
  $st^{(u, v)}_i$ splits to $st^{(u)}_i$ and $st^{(v)}_i$\;
  \While{$i$ $\geq$ 0} {
    \For{y $\in$ $\{adj^{nt}_i(x)$ $|$ x $\in$ $st^{(u)}_i\}$ }{
      \If{(tx, ty) reconnects $st^{(u)}_i$ and $st^{(v)}_i$}{
        insert (tx, ty) as a level-$i$ tree edge\;
        merge $st^{(u)}_i$ and $st^{(v)}_i$ into $st^{(u, v)}_i$\;
        push down level-$i$ tree edges and/or non-tree edges\;
        \Return\;
      }
    }
    $i$ $--$\\
  }
\end{algorithm2e}

\section{Major Data Structures}
\label{sec:maintaining_ds}

We classify existing data structures into three categories based on
the type of spanning trees they are maintaining.  In the first
category , the D-tree~\cite{dtree} is a spanning tree without levels,
and link-cut tree~\cite{Sleator_Tarjan_81} 
maintains a set of splay trees~\cite{Sleator_Tarjan_85} to present 
spanning trees without levels.
In the second category, HK~\cite{henzinger1999randomized},
HKS~\cite{David_HKvariant} (a simplified version of HK), and
HDT~\cite{HDT} use Euler Tour trees (ET-trees)
~\cite{henzinger1999randomized} to maintain cumulative spanning
trees. The third category, structural trees ~\cite{thorup2000near,
  wulff2013faster}, local trees~\cite{thorup2000near,
  wulff2013faster}, and lazy local trees~\cite{thorup2000near,
  wulff2013faster, huang2017fully}, maintain height bounded recursive
spanning trees.  When describing the maintenance of data
structures, we focus on discussing operations for inserting and
deleting tree edges as operations for non-tree edges are
trivial. Inserting and deleting non-tree edges do not change spanning
trees, we simply update non-tree neighbors of the nodes.

\subsection{D-tree}

Tree-data structures with nodes that have a high fanout tend to be
shallow, resulting in efficient runtime for answering connectivity
queries.  D-trees apply this principle to spanning trees by minimizing
parameter $S_d$, the sum of distances between the nodes in the tree
and the root node. Since D-trees process connectivity queries by
traversing from query nodes to root nodes (to check if the query nodes
are located in the same tree), minimizing $S_d$ results in a low
average runtime.

Constructing a
BFS tree results in a spanning tree with an optimal value for $S_d$.
However, maintaining optimal BFS-trees is too expensive for large
dynamic graphs, so D-trees employ heuristics to keep the value of
$S_d$ low (for details, see~\cite{dtree}). A side effect of having a
tree with a low value for $S_d$ is that deleting a tree edge $(u,v)$
usually splits a spanning tree into a large tree $st^{(v)}$,
containing $v$, and a small tree $st^{(u)}$, containing $u$ (w.l.o.g.
we assume the large tree contains $v$). After deleting the edge
$(u,v)$, we traverse $st^{(u)}$ to search for a replacement (non-tree)
edge.  If we find one we can reconnect $st^{(v)}$ and $st^{(u)}$, i.e.,
the nodes in $st^{(v)}$ and $st^{(u)}$ are still connected in the
graph. Usually, $st^{(u)}$ contains a very small number of nodes,
often fewer than ten~\cite{dtree}.  If multiple replacement edges
exist, we choose the one that re-attaches $st^{(u)}$ to the node that is
closest to the root of $st^{(v)}$ to keep the tree as shallow as
possible. In order to reconnect $st^{(u)}$, we may have to reroot
it. This may also be the case when inserting a new tree edge
connecting two previously unconnected components~\cite{dtree}.

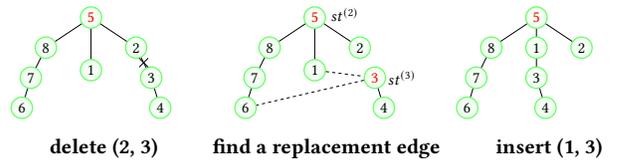
\begin{figure}[htb!] 
  \begin{subfigure}[b]{0.30\columnwidth}
  \scalebox{0.40}{  
    \begin{tikzpicture}[roundnode/.style={circle, draw=green!60, fill=green!5, very thick, minimum size=4mm}]
      \node[roundnode] (5)  at (0, 0)  {\textcolor{red}{5}};
      \node[roundnode] (8)  at (-1.5, -1.0) {8};
      \node[roundnode] (1)  at (0.0, -1.75)  {1};
      \node[roundnode] (2)  at (1.5, -1.0)  {2};
      
      \node[roundnode] (7)  at (-2.0, -2.0) {7};
      \node[roundnode] (3)  at (2.0, -2.0)  {3};
      
      \node[roundnode] (6)  at (-2.3, -3.0) {6};
      \node[roundnode] (4)  at (2.3, -3.0) {4};
      
      \draw plot[only marks,mark=x, mark size=6pt] coordinates {(1.75, -1.5)};
      
      \path[-] (2) edge node[above] {} (3);
      \path[-] (3) edge node[above] {} (4);
      \path[-] (2) edge node[left] {} (5);
      \path[-] (1) edge node[left] {} (5);
      
      \path[-] (7) edge node[above] {} (6);
      \path[-] (7) edge node[above] {} (8);
      \path[-] (8) edge node[above] {} (5);
       
    \end{tikzpicture}
  }
  \caption*{delete (2, 3)}
  \end{subfigure}
  \begin{subfigure}[b]{0.38\columnwidth}\centering
  \scalebox{0.40}{  
    \begin{tikzpicture}[roundnode/.style={circle, draw=green!60, fill=green!5, very thick, minimum size=4mm}]
      \node[roundnode] (5)  at (0, 0)  {\textcolor{red}{5}};
      \node[roundnode] (8)  at (-1.5, -1.0) {8};
      \node[roundnode] (1)  at (0.0, -1.75)  {1};
      \node[roundnode] (2)  at (1.5, -1.0)  {2};
      \node () at (1.0, 0.1) {$st^{(2)}$};

      \node[roundnode] (7)  at (-2.0, -2.0) {7};
      \node[roundnode] (3)  at (2.0, -2.0)  {\textcolor{red}{3}};
      \node () at (2.9, -2.0) {$st^{(3)}$};

      \node[roundnode] (6)  at (-2.3, -3.0) {6};
      \node[roundnode] (4)  at (2.3, -3.0) {4};

      \path[-] (3) edge node[above] {} (4);
      \path[-] (2) edge node[left] {} (5);
      \path[-] (1) edge node[left] {} (5);
      
      \path[-, dashed] (1) edge node[left] {} (3);
      \path[-, dashed] (3) edge node[left] {} (6);

      \path[-] (7) edge node[above] {} (6);
      \path[-] (7) edge node[above] {} (8);
      \path[-] (8) edge node[above] {} (5);
       
  \end{tikzpicture}
  }
  \caption*{find a replacement edge}
  \end{subfigure}
  \begin{subfigure}[b]{0.30\columnwidth}
  \scalebox{0.40}{  
    \begin{tikzpicture}[roundnode/.style={circle, draw=green!60, fill=green!5, very thick, minimum size=4mm}]
        \node[roundnode] (5)  at (0, 0)  {\textcolor{red}{5}};
        \node[roundnode] (8)  at (-1.5, -1.0) {8};
        \node[roundnode] (1)  at (0.0, -1.0)  {1};
        \node[roundnode] (2)  at (1.5, -1.0)  {2};
        
        \node[roundnode] (7)  at (-2.0, -2.0) {7};
        \node[roundnode] (3)  at (0.0, -2.0)  {3};
        
        \node[roundnode] (6)  at (-2.3, -3.0) {6};
        \node[roundnode] (4)  at (0.3, -3.0) {4};

        \path[-] (1) edge node[above] {} (3);
        \path[-] (3) edge node[above] {} (4);
        \path[-] (2) edge node[left] {} (5);
        \path[-] (1) edge node[left] {} (5);
        
        \path[-] (7) edge node[above] {} (6);
        \path[-] (7) edge node[above] {} (8);
        \path[-] (8) edge node[above] {} (5);
         
    \end{tikzpicture}
  }
  \caption*{insert (1, 3)}
  \end{subfigure}
  \caption{Operations for delete tree edge (2, 3) in a D-tree
  for component $C$ of Figure~\ref{fig:exam_c_t}.  
  Nodes with red labels are root nodes.}
  \label{fig:exam_dtree}
\end{figure}

\begin{example} \label{exam:dtree_delete}
  In Figure~\ref{fig:exam_dtree}, after tree edge (2, 3) is deleted,
  the D-tree is split to $st^{(2)}$ containing node 2 and $st^{(3)}$
  containing node 3.  The small tree $st^{(3)}$ is traversed to search
  for non-tree edges. The non-tree edges (1, 3) and (3, 6) reconnect
  $st^{(2)}$ and $st^{(3)}$.  Edge (1, 3) attaches $st^{(3)}$ to node
  1 while edge (3, 6) attaches $st^{(3)}$ to node 6. The non-tree edge
  (1, 3) is selected as inserting (1, 3) into the D-tree maintains a
  smaller $S_d$.
\end{example}

\subsection{Link-cut Tree}
\label{sec:lct}

Link-cut trees (LCTs) are a data structure that dynamically adjusts to the
workload, bringing down the amortized costs of the operations applied to
it~\cite{Sleator_Tarjan_81}. Rather than balancing a tree, an LCT partitions
it into separate paths, called \emph{preferred paths}, consisting of
\emph{preferred edges}. Nodes that are not part of a preferred path are called
\emph{isolated} (see Figure~\ref{fig:exam_lct}(a) for examples of preferred
paths and an isolated node 1). Every preferred path is stored as a splay tree
(binary search tree)~\cite{Sleator_Tarjan_85}, in which the nodes on the
preferred path are ordered according to their depth on the path (e.g., see
Figure~\ref{fig:exam_lct}(b) for the splay tree of path 2-3-4). Each isolated
node is the root node of a splay tree only containing itself. If edge $(u,v)$
is a preferred edge, nodes $u$ and $v$ are in the same splay tree, with $u$
being the predecessor of $v$. If $(u,v)$ is a non-preferred edge, $u$ and $v$
are found in different splay trees, here denoted $sp^{(u)}$ and $sp^{(v)}$,
respectively. Then $(u,v)$ is stored as a \emph{directed} pointer, called
\emph{path\_pointer}, from the root node of $sp^{(v)}$ to $u$ (in
Figure~\ref{fig:exam_lct}(b), the edges $(1,5)$ and $(3,5)$ are
path\_pointers).

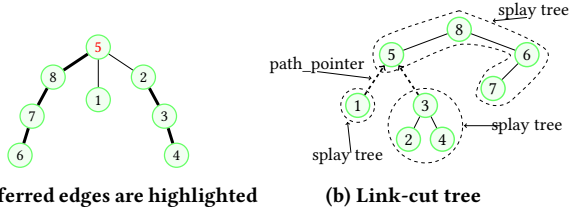
\begin{figure}[htb!] 
  \begin{subfigure}[b]{0.50\columnwidth} \centering
    \scalebox{0.4}{  
    \begin{tikzpicture}
      \node[roundnode] (5)  at (0, 0)  {\textcolor{red}{5}};
      \node[roundnode] (8)  at (-1.5, -1.0) {8};
      \node[roundnode] (1)  at (0.0, -1.75)  {1};
      \node[roundnode] (2)  at (1.5, -1.0)  {2};
      
      \node[roundnode] (7)  at (-2.2, -2.3) {7};
      \node[roundnode] (3)  at (2.2, -2.3)  {3};
      
      \node[roundnode] (6)  at (-2.6, -3.6) {6};
      \node[roundnode] (4)  at (2.6, -3.6) {4};
      
      \path[-, line width=1mm] (2) edge node[right] {} (3);
      \path[-, line width=1mm] (3) edge node[right] {} (4);
      \path[-] (2) edge node[above] {} (5);
      \path[-] (1) edge node[left] {} (5);
      
      \path[-, line width=1mm] (7) edge node[left] {} (6);
      \path[-, line width=1mm] (7) edge node[left] {} (8);
      \path[-, line width=1mm] (8) edge node[above] {} (5);
    \end{tikzpicture}
    }
    \caption{Preferred edges are highlighted}
    \label{fig:et_example}
  \end{subfigure}  
  \begin{subfigure} [b]{0.44\columnwidth} \centering
    \scalebox{0.45} {
      \begin{tikzpicture}[roundnode/.style={circle, draw=green!60, fill=green!5, very thick, minimum size=7mm}]
        \node[roundnode] (8)  at (0, 0.25)  {8};
        \node[roundnode] (5)  at (-2.0, -0.5)  {5};
        \node[roundnode] (6)  at (2.0, -0.5)  {6};
        \node[roundnode] (7)  at (1.0, -1.5)  {7};
        
        \path[-] (5) edge  (8);
        \path[-] (6) edge  (8);
        \path[-] (6) edge  (7);
        
        \draw [dashed] plot [smooth] coordinates {(0, 0.8) (-2.3, 0.0) (-2.5, -0.8) (-2.0, -1.0) 
        (0.0, -0.2) (1.5, -0.5) (0.8, -1.0) (0.5, -1.7) (1.0, -2.0) (1.6, -1.8) (2.6, -0.4) (0, 0.8)} ;
        \node at (2.2, 0.8) {splay tree};
        \draw[->] (2.0, 0.6) -- (1.0, 0.5);
        
        \node[roundnode] (1)  at (-3.0, -2.0)  {1};
        \draw [dashed] (-3.0, -2.0) circle [radius=5mm];
        \node at (-3.3, -3.5) {splay tree};
        \draw [->] (-3.2, -3.4) -- (-3.3, -2.5);
        
        \node at (-4.2, -0.8) {path\_pointer};
        \draw [->] (-4.2, -1.0) -- (-2.6, -1.2);
         
        \node[roundnode] (3)  at (-1.0, -2.0)  {3};
        \node[roundnode] (2)  at (-1.5, -3.0)  {2};
        \node[roundnode] (4)  at (-0.5, -3.0)  {4};
        \draw [dashed] (-1.0, -2.6) circle [radius=11mm];
        \node at (2.0, -2.6) {splay tree};
        \draw [->] (1.2, -2.6) -- (0.15, -2.6);
        
        \path[-] (3) edge  (2);
        \path[-] (3) edge  (4);
        
        \path[->, dashed, line width = 0.5mm] (1) edge  (5);
        \path[->, dashed, line width = 0.5mm] (3) edge  (5);
        
      \end{tikzpicture}
    }
    \caption{Link-cut tree}
  \end{subfigure}
  
  \caption{Consider the spanning tree in
      Figure~\ref{fig:exam_dtree}. Two preferred paths and the
      isolated node 1 of the spanning tree are stored as splay trees
      in the Link/cut tree.}
  \label{fig:exam_lct}
\end{figure}

We now turn to representing spanning trees with LCTs. Every non-leaf node in
the spanning tree has one preferred child, connected via a preferred
edge. However, an arbitrary tree edge of the spanning tree can either be a
preferred edge of a non-preferred edge.  Deleting a tree edge $(u, v)$ splits
the spanning tree $st^{(u, v)}$ into $st^{(u)}$ and $st^{(v)}$. The link-cut
tree $LCT^{(u, v)}$ for $st^{(u, v)}$ splits into two link-cut trees,
$LCT^{(u)}$ containing $u$ and $LCT^{(v)}$ containing $v$.  There are no
procedures in LCTs optimized to search for replacement edges during the
deletion of tree edges, though. The directedness of path\_pointers makes it
hard to efficiently search either $LCT^{(u)}$ or $LCT^{(v)}$ in a top-down
fashion, as we do not see subtrees connected via path\_pointers from a parent
node. Without auxiliary data structures, we have to traverse all the nodes to
determine the nodes contained in a specific LCT. Additionally, inserting a
replacement edge can merge two preferred paths. $W.l.o.g,$ assume that
$(t\_u, t\_v)$ is the replacement edge and the preferred path containing node
$t_v$ is attached to the preferred path containing $t_u$. All ancestors of
$t\_v$ in the preferred path become descendents of $t\_u$ and have to be
reordered.

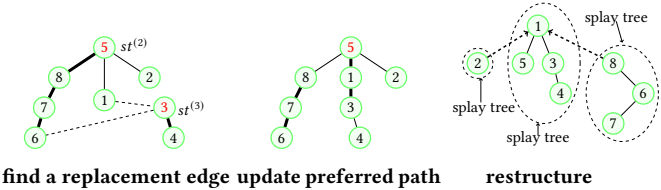
\begin{figure}[htb!] 
  \begin{subfigure}[b]{0.36\columnwidth}\centering
  \scalebox{0.40}{  
    \begin{tikzpicture}[roundnode/.style={circle, draw=green!60, fill=green!5, very thick, minimum size=4mm}]
      \node[roundnode] (5)  at (0, 0)  {\textcolor{red}{5}};
      \node[roundnode] (8)  at (-1.5, -1.0) {8};
      \node[roundnode] (1)  at (0.0, -1.75)  {1};
      \node[roundnode] (2)  at (1.5, -1.0)  {2};
      \node () at (1.0, 0.1) {$st^{(2)}$};

      \node[roundnode] (7)  at (-2.0, -2.0) {7};
      \node[roundnode] (3)  at (2.0, -2.0)  {\textcolor{red}{3}};
      \node () at (2.9, -2.0) {$st^{(3)}$};

      \node[roundnode] (6)  at (-2.3, -3.0) {6};
      \node[roundnode] (4)  at (2.3, -3.0) {4};

      \path[-, line width=1mm] (3) edge node[above] {} (4);
      \path[-] (2) edge node[left] {} (5);
      \path[-] (1) edge node[left] {} (5);
      
      \path[-, dashed] (1) edge node[left] {} (3);
      \path[-, dashed] (3) edge node[left] {} (6);

      \path[-, line width=1mm] (7) edge node[above] {} (6);
      \path[-, line width=1mm] (7) edge node[above] {} (8);
      \path[-, line width=1mm] (8) edge node[above] {} (5);
       
  \end{tikzpicture}
  }
  \caption*{find a replacement edge}
  \end{subfigure}
  \begin{subfigure}[b]{0.32\columnwidth}\centering
  \scalebox{0.40}{  
    \begin{tikzpicture}[roundnode/.style={circle, draw=green!60, fill=green!5, very thick, minimum size=4mm}]
        \node[roundnode] (5)  at (0, 0)  {\textcolor{red}{5}};
        \node[roundnode] (8)  at (-1.5, -1.0) {8};
        \node[roundnode] (1)  at (0.0, -1.0)  {1};
        \node[roundnode] (2)  at (1.5, -1.0)  {2};
        
        \node[roundnode] (7)  at (-2.0, -2.0) {7};
        \node[roundnode] (3)  at (0.0, -2.0)  {3};
        
        \node[roundnode] (6)  at (-2.3, -3.0) {6};
        \node[roundnode] (4)  at (0.3, -3.0) {4};

        \path[-, line width=1mm] (1) edge node[above] {} (3);
        \path[-] (3) edge node[above] {} (4);
        \path[-] (2) edge node[left] {} (5);
        \path[-, line width=1mm] (1) edge node[left] {} (5);
        
        \path[-, line width=1mm] (7) edge node[above] {} (6);
        \path[-, line width=1mm] (7) edge node[above] {} (8);
        \path[-] (8) edge node[above] {} (5);

    \end{tikzpicture}
  }
  \caption*{update preferred path}
  \end{subfigure}
  \begin{subfigure}[b]{0.29\columnwidth}
    \scalebox{0.40}{  
      \begin{tikzpicture}[roundnode/.style={circle, draw=green!60, fill=green!5, very thick, minimum size=7mm}]
        \node[roundnode] (1)  at (-2.0, 0.25)  {1};
        
        \node[roundnode] (5)  at (-2.5, -1.0)  {5};
        \node[roundnode] (3)  at (-1.5, -1.0)  {3};
        \node[roundnode] (4)  at (-1.25, -2.0)  {4};
        
        \draw [dashed] (-1.8, -1.0) ellipse (1.2cm and 2.0cm);
        \node at (-2.0, -3.5) {splay tree};
        \draw [->] (-2.0, -3.4) -- (-2.0, -3.0);
        
        \node[roundnode] (2)  at (-4.0, -1.0)  {2};
        \draw [dashed] (-4.0, -1.0) circle [radius=5mm];
        \node at (-3.8, -2.5) {splay tree};
        \draw [->] (-3.9, -2.3) -- (-3.9, -1.5);
        
        \node[roundnode] (8)  at (0.5, -1.0)  {8};

        \node[roundnode] (6)  at (1.5, -2.0)  {6};
        \node[roundnode] (7)  at (0.5, -3.0)  {7};
        
        \path[->, dashed, line width = 0.5mm] (2) edge  (1);
        \path[-] (6) edge  (8);
        \path[-] (6) edge  (7);
        
        \path[->, dashed, line width = 0.5mm] (8) edge  (1);
        
        \draw [dashed] (0.8, -2.0) ellipse (1.2cm and 1.6cm);
        \node at (0.6, 0.5) {splay tree};
        \draw [->] (0.6, 0.4) -- (0.8, -0.3);

        \path[-] (3) edge  (4);
        
        \path[-] (1) edge  (5);
        \path[-] (3) edge  (1);
        
      \end{tikzpicture}
    }
    \caption*{restructure}
  \end{subfigure}
  \caption{Delete tree edge (2, 3). Select edge (1, 3) as the replacement edge and insert (1, 3).}
  \label{fig:exam_lct_update}
\end{figure}

\subsection{ET-trees for HK, HKS and HDT}
\label{sec:et}

HK~\cite{henzinger1999randomized} and HDT~\cite{HDT} use ET-trees, implemented
with randomized search trees \cite{seidel1996randomized}, that have no total
orders of nodes to maintain cumulative spanning trees for every
level.\footnote{Note: 
for unweighted graphs each spanning tree is also
 a minimum spanning tree.} Merging two balanced binary trees takes $O(\log n)$ time.

An Euler Tour visits all nodes of a spanning tree, starting and ending
at the same node (see the example shown
Figure~\ref{fig:et_example}) \cite{tarjan1985efficient, David_HKvariant}. 
Each edge is visited twice and each node, except leaf nodes,
is visited multiple times, depending on the
number of tree neighbors of the node. We write $u_i$ to denote the
$i$-th occurrence of node $u$ in the Euler Tour.  For an edge $(u, v)$
of the spanning tree, we write <$u_x, v_y$> and <$v_p, u_q$> to denote
the first and second visits of $u$ and $v$ by the Euler tour,
respectively. Thus, edge $(u, v)$ is associated with the four
occurrences $u_x$, $v_y$, $v_p$ and $u_q$ in the Euler tour.  Consider
a spanning tree with $n$ nodes, an Euler tour includes $2n - 1$
occurrences of the nodes.  The ET-tree is a transformation of the
Euler Tour into a balanced binary tree, such that each occurrence
$u_i$ in the Euler Tour is associated with a node with key $u_i$ in
the ET-tree.  Note that the sequence of occurrences of nodes in
the inorder traversal of this balanced binary tree is equal to the
Euler Tour.  The ET-tree does not preserve the tree edges of the
spanning tree and hence ET-trees must maintain tree edge information
explicitly (as attribute values of the nodes in the ET-tree; 
see Section \ref{sec:imp_hks}).

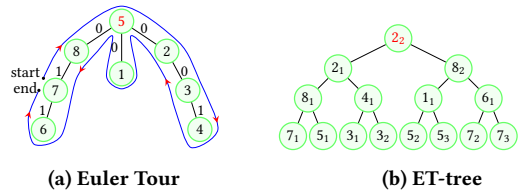
\begin{figure}[htb!] 
  \begin{subfigure}[b]{0.49\columnwidth} \centering
    \scalebox{0.40}{  
    \begin{tikzpicture}
      \node[roundnode] (5)  at (0, 0)  {\textcolor{red}{5}};
      \node[roundnode] (8)  at (-1.5, -1.0) {8};
      \node[roundnode] (1)  at (0.0, -1.75)  {1};
      \node[roundnode] (2)  at (1.5, -1.0)  {2};
      
      \node[roundnode] (7)  at (-2.2, -2.3) {7};
      \node[roundnode] (3)  at (2.2, -2.3)  {3};
      
      \node[roundnode] (6)  at (-2.6, -3.6) {6};
      \node[roundnode] (4)  at (2.6, -3.6) {4};

      \node (7_1)  at (-2.6, -1.9) {};
      \node (s)  at (-3.2, -1.7) {start};
      \node (8_1)  at (-1.9, -0.7) {};
      \node (5_1)  at (0.0, 0.5) {};
      \node (2_1)  at (1.9, -0.7) {};
      \node (3_1)  at (2.6, -2.0)  {};
      \node (4_1)  at (3.2, -3.6) {};
      \node (3_2)  at (1.5, -2.0)  {};
      \node (2_2)  at (0.9, -1.0)  {};
      \node (5_2)  at (0.3, -0.5) {};
      \node (1_1)  at (0.5, -1.85) {};
      \node (5_3)  at (-0.3, -0.5) {};
      \node (8_2)  at (-0.8, -1.0) {};
      \node (7_2)  at (-1.5, -1.9) {};
      \node (6_1)  at (-2.0, -3.8) {};
      \node (7_3)  at (-2.75, -2.3) {};
      \node (e)  at (-3.2, -2.2) {end};
      
      \draw[blue] plot [smooth] coordinates {(7_1) (8_1) (5_1) 
      (2_1) (3_1) (4_1) (2.3, -4.2) (1.8, -3.0) (3_2) (2_2) 
      (5_2) (1_1) (0.0, -2.25) (-0.5, -1.85) (5_3) (8_2) (7_2) (6_1) 
      (-2.9, -4.1) (-3.1, -3.2) (7_3)}  
      [arrow inside={end=stealth,opt={red,scale=2}}
      {0.05, 0.35,0.50,0.80, 0.97}];

      \draw [fill=black] (7_1) circle (1pt);
      \draw [fill=black] (7_3) circle (1pt);
      
      \path[-] (2) edge node[right] {0} (3);
      \path[-] (3) edge node[right] {1} (4);
      \path[-] (2) edge node[above] {0} (5);
      \path[-] (1) edge node[left] {0} (5);
      
      \path[-] (7) edge node[left] {1} (6);
      \path[-] (7) edge node[left] {1} (8);
      \path[-] (8) edge node[above] {0} (5);
    \end{tikzpicture}
    }
    \caption{Euler Tour}
    \label{fig:et_example}
  \end{subfigure}  
  \begin{subfigure} [b]{0.49\columnwidth}
    \scalebox{0.4} {
      \begin{tikzpicture}[roundnode/.style={circle, draw=green!60, fill=green!5, very thick, minimum size=7mm}]
        \node[roundnode] (2_2)  at (0, 0)  {\textcolor{red}{$2_2$}};
        
        \node[roundnode] (2_1)  at (-2.0, -1.0)  {$2_1$};
        \node[roundnode] (8_2)  at (2.0, -1.0)  {$8_2$};
        
        \path[-] (2_2) edge  (2_1);
        \path[-] (2_2) edge  (8_2);
        
        \node[roundnode] (8_1)  at (-3.0, -2.0)  {$8_1$};
        \node[roundnode] (4_1)  at (-1.0, -2.0)  {$4_1$};
        \node[roundnode] (1_1)  at (1.0, -2.0)  {$1_1$};
        \node[roundnode] (6_1)  at (3.0, -2.0)  {$6_1$};
        
        \path[-] (2_1) edge  (8_1);
        \path[-] (2_1) edge  (4_1);
        \path[-] (8_2) edge  (1_1);
        \path[-] (8_2) edge  (6_1);
        
        \node[roundnode] (7_1)  at (-3.5, -3.25)  {$7_1$};
        \node[roundnode] (5_1)  at (-2.5, -3.25)  {$5_1$};
        \node[roundnode] (3_1)  at (-1.5, -3.25)  {$3_1$};
        \node[roundnode] (3_2)  at (-0.5, -3.25)  {$3_2$};
        \node[roundnode] (5_2)  at (0.5, -3.25)  {$5_2$};
        \node[roundnode] (5_3)  at (1.5, -3.25)  {$5_3$};
        \node[roundnode] (7_2)  at (2.5, -3.25)  {$7_2$};
        \node[roundnode] (7_3)  at (3.5, -3.25)  {$7_3$};
        
        \path[-] (8_1) edge  (7_1);
        \path[-] (8_1) edge  (5_1);
        \path[-] (4_1) edge  (3_1);
        \path[-] (4_1) edge  (3_2);
        \path[-] (1_1) edge  (5_2);
        \path[-] (1_1) edge  (5_3);
        \path[-] (6_1) edge  (7_2);
        \path[-] (6_1) edge  (7_3);    
      \end{tikzpicture}
    }
    \caption{ET-tree}
  \end{subfigure}
  
  \caption{An Euler Tour of a level-0 cumulative 
  spanning tree $_c\mathcal{ST}_0$ 
    in Figure~\ref{fig:exam_cst} starting at node 7 is 
  $7_1$, $8_1$, $5_1$, $2_1$, $3_1$, $4_1$, $3_2$, $2_2$, 
  $5_2$, $1_1$, $5_3$, $8_2$, $7_3$, $6_1$, $7_3$ 
  (first visiting node 7 and then 
  nodes 8, 5, 2, 3, 4, 3 and so on). The ET-tree 
  is shown on the right hand side.}
  \label{fig:exam_et}
\end{figure}

\begin{example}
  In Figure~\ref{fig:exam_et}, the spanning tree has 8 nodes and the
  ET-tree has 15 nodes. In the ET-tree, the edge between $2_2$ and
  $8_2$ does not exist in the spanning tree. The edge $(7, 8)$ in the
  spanning tree is mapped to four tree nodes with keys $7_1$, $8_1$,
  $8_2$, $7_2$.
\end{example}

Deleting a level-$i$ tree edge $(u, v)$ splits the spanning tree
$st_i^{(u, v)}$ into $st_i^{(u)}$ and $st_i^{(v)}$ and hence splits
the Euler Tour of $st_i^{(u, v)}$. Let $st_i^{(u)}$ be the smaller
tree.  We traverse the nodes in the ET-tree of $st_i^{(u)}$ to find a
replacement edge.  If there are multiple replacement edges, any of
them can be selected. Assume non-tree edge $(tx, ty)$ reconnects
$tx \in st_i^{(u)}$ and $ty \in st_i^{(v)}$.  The ET-tree for
$st_i^{(u)}$ and the ET-tree for $st_i^{(v)}$ are merged through
$(tx, ty)$, which means $(tx, ty)$ must be inserted as a tree edge.
The algorithm to merge two ET-trees through $(tx, ty)$ requires that
the Euler tours of $st_i^{(u)}$ and $st_i^{(v)}$ start at node $tx$
and node $ty$, respectively.  If this condition is not satisfied, we
restructure the ET-trees before merging them. 

\begin{figure}[htb!] 
  \begin{subfigure}[b]{0.3\columnwidth}\centering
    \scalebox{0.35}{
      \begin{tikzpicture}
      \node[roundnode] (5)  at (0, 0)  {\textcolor{red}{5}};
      \node[roundnode] (8)  at (-1.3, -0.8) {8};
      \node[roundnode] (1)  at (0.0, -1.75)  {1};
      \node[roundnode] (2)  at (1.3, -0.8)  {2};
      
      \node[roundnode] (7)  at (-1.9, -1.9) {7};
      \node[roundnode] (3)  at (2.0, -2.0)  {3};
      
      \node[roundnode] (6)  at (-2.3, -3.0) {6};
      \node[roundnode] (4)  at (2.4, -3.3) {4};

      \node (7_1)  at (-2.4, -1.7) {};
      \node (s)  at (-2.9, -1.5) {start};
      \node (8_1)  at (-1.8, -0.6) {};
      \node (5_1)  at (0.0, 0.5) {};
      \node (2_1)  at (1.8, -0.6) {};
      \node (3_1)  at (2.5, -1.8)  {};
      \node (4_1)  at (2.9, -3.1) {};
      \node (3_2)  at (1.5, -2.0)  {};
      \node (2_2)  at (0.8, -1.0)  {};
      \node (5_2)  at (0.3, -0.5) {};
      \node (1_1)  at (0.5, -1.85) {};
      \node (5_3)  at (-0.3, -0.5) {};
      \node (8_2)  at (-0.9, -1.0) {};
      \node (7_2)  at (-1.4, -1.9) {};
      \node (6_1)  at (-1.8, -3.2) {};
      \node (7_3)  at (-2.5, -2.0) {};
      \node (e)  at (-3.0, -2.0) {end};
      
      \draw[blue] plot [smooth]
      coordinates {(7_1) (8_1) (5_1) (2_1)  (1.8, -1.1) (1.4, -1.4) (2_2)
      (5_2) (1_1) (0.0, -2.25) (-0.5, -1.85) (5_3) (8_2) (7_2) (6_1)  (-2.4, -3.5) 
      (-2.8, -3.2) (7_3)}
      [arrow inside={end=stealth,opt={red,scale=2}}{0.05, 0.35,0.55,0.70, 0.95}];

      \draw [fill=black] (7_1) circle (1pt);
      \draw [fill=black] (7_3) circle (1pt);
      
      \path[-] (3) edge node[right] {1} (4);
      \path[-] (2) edge node[above] {0} (5);
      \path[-] (1) edge node[left] {0} (5);
      
      \path[-] (7) edge node[right] {1} (6);
      \path[-] (7) edge node[right] {1} (8);
      \path[-] (8) edge node[above] {0} (5);

      \path[-, dashed] (1) edge node[right] {} (3);
      \path[-, dashed] (3) edge node[right] {} (6);
      
      \draw[blue] plot [smooth] coordinates {(3_1) (4_1) (2.8, -3.6)  (2.3, -3.7) 
      (1.8, -3.2) (3_2) (2.0, -1.5) (3_1)}
      [arrow inside={end=stealth,opt={red,scale=2}}{0.25}];
    \end{tikzpicture}
    }
    \caption*{delete (2, 3)}
  \end{subfigure}
  \begin{subfigure}[b]{0.33\columnwidth}
    \scalebox{0.35}{
    \begin{tikzpicture}
      \node[roundnode] (5)  at (0, 0)  {\textcolor{red}{5}};
      \node[roundnode] (8)  at (-1.3, -0.8) {8};
      \node[roundnode] (1)  at (0.0, -1.75)  {1};
      \node[roundnode] (2)  at (1.3, -0.8)  {2};
      
      \node[roundnode] (7)  at (-1.9, -1.9) {7};
      \node[roundnode] (3)  at (2.0, -2.0)  {3};
      
      \node[roundnode] (6)  at (-2.3, -3.0) {6};
      \node[roundnode] (4)  at (2.4, -3.3) {4};
      
      \node (6_1)  at (-2.75, -2.8) {};
      \node (7_1)  at (-2.4, -1.7) {};
      \node (8_1)  at (-1.8, -0.6) {};
      \node (5_1)  at (0.0, 0.5) {};
      \node (2_1)  at (1.8, -0.6) {};
      \node (3_1)  at (2.5, -1.8)  {};
      \node (4_1)  at (2.9, -3.1) {};
      \node (3_2)  at (1.5, -2.0)  {};
      \node (2_2)  at (0.8, -1.0)  {};
      \node (5_2)  at (0.3, -0.5) {};
      \node (1_1)  at (0.5, -1.85) {};
      \node (5_3)  at (-0.3, -0.5) {};
      \node (8_2)  at (-0.85, -1.0) {};
      \node (7_2)  at (-1.4, -1.9) {};
      \node (6_2)  at (-1.8, -3.2) {};
      
      \draw[blue] plot [smooth]
      coordinates {(6_1) (7_1) (8_1) (5_1) (2_1)  (1.8, -1.1) (1.4, -1.4) (2_2)
      (5_2) (1_1) (0.0, -2.25) (-0.5, -1.85) (5_3) (8_2) (7_2) (6_2)  (-2.4, -3.5) 
      (-2.8, -3.1) }
      [arrow inside={end=stealth,opt={red,scale=2}}{0.05, 0.35,0.55,0.70, 0.95}];

      \draw [fill=black] (6_1) circle (1pt);
      \node (s)  at (-3.2, -2.5) {start};
      
      \draw [fill=black] (-2.8, -3.1) circle (1pt);
      \node (e)  at (-3.25, -3.2) {end};
      
      \path[-] (3) edge node[right] {1} (4);
      \path[-] (2) edge node[above] {0} (5);
      \path[-] (1) edge node[left] {0} (5);
      
      \path[-] (7) edge node[right] {1} (6);
      \path[-] (7) edge node[right] {1} (8);
      \path[-] (8) edge node[above] {0} (5);

      \path[-] (3) edge node[right] {} (6);
      
      \draw[blue] plot [smooth] coordinates {(3_1) (4_1) (2.8, -3.6)  (2.3, -3.7) 
      (1.8, -3.2) (3_2) (2.0, -1.5) (3_1)}
      [arrow inside={end=stealth,opt={red,scale=2}}{0.25}];
    \end{tikzpicture}
    }
    \caption*{restructure}
  \end{subfigure}
  \begin{subfigure}[b]{0.33\columnwidth}
    \scalebox{0.35}{
    \begin{tikzpicture}[roundnode/.style={circle, draw=green!60, fill=green!5, very thick, minimum size=4mm}]
      \node[roundnode] (5)  at (0, 0)  {\textcolor{red}{5}};
      \node[roundnode] (8)  at (-1.5, -1.0) {8};
      \node[roundnode] (1)  at (0.0, -1.75)  {1};
      \node[roundnode] (2)  at (1.5, -1.0)  {2};
      
      \node[roundnode] (7)  at (-2.2, -2.0) {7};
      \node[roundnode] (6)  at (-2.5, -3.0) {6};
      \node[roundnode] (3)  at (-1.1, -3.0)  {3};
      \node[roundnode] (4)  at (0.3, -3.0) {4};
      
      \node (7_1)  at (-2.5, -1.7) {};
      
      \node (8_1)  at (-1.8, -0.7) {};
      \node (5_1)  at (0.0, 0.5) {};
      \node (2_1)  at (1.8, -0.7) {};
      \node (6_1)  at (-2.95, -3.3) {};
      \node (4_1)  at (0.3, -2.5) {};
      
      \node (2_2)  at (1.0, -1.0)  {};
      \node (5_2)  at (0.3, -0.5) {};
      \node (1_1)  at (0.5, -1.85) {};
      \node (5_3)  at (-0.3, -0.5) {};
      \node (8_2)  at (-1.0, -1.0) {};
      \node (7_2)  at (-1.6, -1.9) {};
      \node (6_2)  at (-2.0, -2.5) {};
      
      \node (3_1)  at (-1.1, -2.5)  {};
      \node (3_2)  at (-1.5, -3.5)  {}; 
      \node (7_3)  at (-2.4, -2.0) {};
      \node (s)  at (-0.6, -2.3) {start};
      \draw [fill=black] (3_1) circle (1pt);
      
      \node (e)  at (-1.6, -2.3) {end};
      
      \draw [fill=black] (-1.3, -2.5) circle (1pt);
      
      \draw[blue] plot [smooth]
      coordinates { (3_1)
      (4_1) (0.7, -2.9) (0.4, -3.4) (3_2) (6_1) (7_1) 
      (8_1) (5_1) (2_1)  (2.0, -1.2) 
      (1.5, -1.5) (2_2) (5_2) (1_1) (0.0, -2.25) (-0.5, -1.85) 
      (5_3) (8_2) (7_2) (6_2) (-1.3, -2.5)}
      [arrow inside={end=stealth,opt={red,scale=2}}{0.05, 0.35, 
      0.55,0.70, 0.95}];

      \path[-] (3) edge node[above] {1} (4);
      \path[-] (2) edge node[right] {0} (5);
      \path[-] (1) edge node[left] {0} (5);
      
      \path[-] (7) edge node[left] {1} (6);
      \path[-] (7) edge node[left] {1} (8);
      \path[-] (8) edge node[left] {0} (5);
      
      \path[-] (3) edge node[above] {0} (6);
    \end{tikzpicture}
    }
    \caption*{Insert (3, 6)}
  \end{subfigure}

  \begin{subfigure}[b]{0.33\columnwidth}
    \scalebox{0.35}{
      \begin{tikzpicture}[roundnode/.style={circle, draw=green!60, fill=green!5, very thick, minimum size=7mm}]
        \node[roundnode] (2_1)  at (0.5, 0)  {$2_1$};
        \node[roundnode] (8_2)  at (2.0, -1.0)  {$8_2$};
        \path[-] (2_1) edge  (8_2);
        
        \node[roundnode] (8_1)  at (-1.0, -1.0)  {$8_1$};

        \node[roundnode] (1_1)  at (1.0, -2.0)  {$1_1$};
        \node[roundnode] (6_1)  at (3.0, -2.0)  {$6_1$};
        
        \path[-] (8_2) edge  (1_1);
        \path[-] (8_2) edge  (6_1);
        \path[-] (2_1) edge  (8_1);
        
        \node[roundnode] (7_1)  at (-1.5, -2.0)  {$7_1$};
        \node[roundnode] (5_1)  at (-0.5, -2.0)  {$5_1$};

        \node[roundnode] (5_2)  at (0.5, -3.25)  {$5_2$};
        \node[roundnode] (5_3)  at (1.5, -3.25)  {$5_3$};
        \node[roundnode] (7_2)  at (2.5, -3.25)  {$7_2$};
        \node[roundnode] (7_3)  at (3.5, -3.25)  {$7_3$};
        
        \path[-] (8_1) edge  (7_1);
        \path[-] (8_1) edge  (5_1);

        \path[-] (1_1) edge  (5_2);
        \path[-] (1_1) edge  (5_3);
        \path[-] (6_1) edge  (7_2);
        \path[-] (6_1) edge  (7_3);    
        
        \node[roundnode] (4_1)  at (5.0, -2.0)  {$4_1$};
        \node[roundnode] (3_1)  at (4.5, -3.25)  {$3_1$};
        \node[roundnode] (3_2)  at (5.5, -3.25)  {$3_2$};
        
        \path[-] (4_1) edge  (3_1);
        \path[-] (4_1) edge  (3_2);
      \end{tikzpicture}
    }
    \caption*{break ET-tree}
  \end{subfigure}
  \begin{subfigure}[b]{0.33\columnwidth}
    \scalebox{0.30}{
    \begin{tikzpicture}[roundnode/.style={circle, draw=green!60, fill=green!5, very thick, minimum size=7mm}]
      \node[roundnode] (2_1)  at (0.5, 0)  {$2_1$};
      \node[roundnode] (8_2)  at (2.0, -1.0)  {$8_2$};
      \path[-] (2_1) edge  (8_2);
      
      \node[roundnode] (8_1)  at (-1.0, -1.0)  {$8_1$};
      \node[roundnode] (1_1)  at (1.0, -2.0)  {$1_1$};
      \node[roundnode] (7_2)  at (3.0, -2.0)  {$7_2$};
      
      \path[-] (8_2) edge  (1_1);
      \path[-] (8_2) edge  (7_2);
      \path[-] (2_1) edge  (8_1);
      
      \node[roundnode] (6_1)  at (-1.5, -2.0)  {$6_1$};
      \node[roundnode] (5_1)  at (-0.5, -2.0)  {$5_1$};
      \node[roundnode] (5_2)  at (0.5, -3.25)  {$5_2$};
      \node[roundnode] (5_3)  at (1.5, -3.25)  {$5_3$};
      \node[roundnode] (7_1)  at (-1.05, -3.25)  {$7_1$};
      \node[roundnode] (6_2)  at (3.5, -3.25)  {$6_2$};
      
      \path[-] (8_1) edge  (6_1);
      \path[-] (8_1) edge  (5_1);
      \path[-] (6_1) edge  (7_1);
      \path[-] (6_2) edge  (7_2);
      
      \path[-] (1_1) edge  (5_2);
      \path[-] (1_1) edge  (5_3);
      
      \node[roundnode] (4_1)  at (5.5, -2.0)  {$4_1$};
      \node[roundnode] (3_1)  at (5.0, -3.25)  {$3_1$};
      \node[roundnode] (3_2)  at (6.0, -3.25)  {$3_2$};
      
      \path[-] (4_1) edge  (3_1);
      \path[-] (4_1) edge  (3_2);
    \end{tikzpicture}
    }
    \caption*{restructure}
  \end{subfigure}
  \begin{subfigure}[b]{0.3\columnwidth}
    \scalebox{0.30}{
      \begin{tikzpicture}[roundnode/.style={circle, draw=green!60, fill=green!5, very thick, minimum size=7mm}]
        \node[roundnode] (2_1)  at (0.5, 0)  {$2_1$};
        \node[roundnode] (8_2)  at (2.0, -1.0)  {$8_2$};
        \path[-] (2_1) edge  (8_2);
        
        \node[roundnode] (8_1)  at (-1.0, -1.0)  {$8_1$};
        \node[roundnode] (1_1)  at (1.0, -2.0)  {$1_1$};
        \node[roundnode] (9_2)  at (3.0, -2.0)  {$7_2$};
         \node[roundnode] (4_1)  at (-1.75, -2.0)  {$4_1$};
        
        \path[-] (8_2) edge  (1_1);
        \path[-] (8_2) edge  (7_2);
        \path[-] (2_1) edge  (8_1);
        \path[-] (8_1) edge  (4_1);
        
        \node[roundnode] (3_1)  at (-2.25, -3.25)  {$3_1$};
        \node[roundnode] (6_1)  at (-1.25, -3.25)  {$6_1$};
        \node[roundnode] (5_1)  at (-0.25, -2.0)  {$5_1$};
        \node[roundnode] (5_2)  at (0.5, -3.25)  {$5_2$};
        \node[roundnode] (5_3)  at (1.5, -3.25)  {$5_3$};
        
        \node[roundnode] (6_2)  at (3.5, -3.25)  {$6_2$};
        \node[roundnode] (3_2)  at (-1.75, -4.5)  {$3_2$};
        \node[roundnode] (7_1)  at (-0.75, -4.5)  {$7_1$};
        \node[roundnode] (3_3)  at (4.0, -4.5)  {$3_3$};
        
        \path[-] (8_1) edge  (5_1);
        \path[-] (6_1) edge  (3_2);
        \path[-] (6_1) edge  (7_1);
        \path[-] (6_2) edge  (7_2);
        \path[-] (1_1) edge  (5_2);
        \path[-] (1_1) edge  (5_3);
        \path[-] (4_1) edge  (3_1);
        \path[-] (4_1) edge  (6_1);
        \path[-] (6_2) edge  (3_3);
    \end{tikzpicture}
    }
    \caption*{merge ET-trees}
  \end{subfigure}
  \caption{Delete the tree edge (2, 3). Select the non-tree
  edge (3, 6) as a replacement edge. Insert (3, 6) 
  as a tree edge.}
  \label{fig:exam_et_delete}
\end{figure}
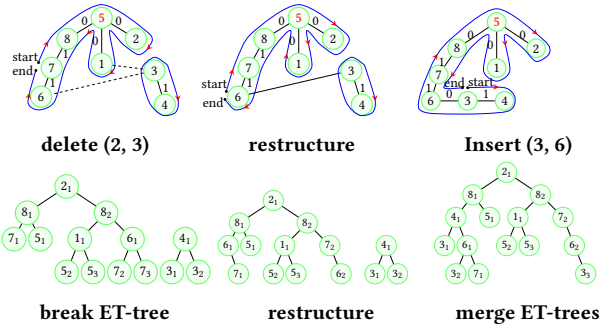

\begin{example}
  Figure~\ref{fig:exam_et_delete} shows the maintenance of ET-tree
  when tree edge (2, 3) is deleted. The smaller ET-tree for
  $st_1^{(3)}$ is traversed.  The non-tree edges (1, 3) and (3, 6) are
  possible replacement edges.  We insert (3, 6) as a tree edge.
  Details of breaking, restructuring and merging ET-trees can be found
  in ~\cite{henzinger1999randomized, David_HKvariant}.
\end{example}


\subsection{Structural Trees and Local Trees}

Structural trees and local trees have bounded heights, which offers
amortized cost guarantees for connectivity queries.  We start out with
the structural tree, which is the baseline model for local
trees~\cite{thorup2000near, wulff2013faster}. 

\subsubsection{Structural Tree}
\label{sec:st}

In a structural tree, nodes of level-$i$ recursive spanning trees are
directly connected to a level-$i$ super node $s_i$ of a recursive
spanning tree.

\begin{definition}[Level-i Structural Tree]
  Given a level-$i$ recursive spanning tree $_r\mathcal{ST}_i$ $=$
  $(_r\mathcal{V}_i,~ E_i)$ and the level-$i$ super node $s_i$ for
  $_r\mathcal{ST}_i$, the level-$i$ structural tree is $T_i$ $=$
  $(V^{T_i}, E^{T_i})$ in which $V^{T^i}$ $=$ $_r\mathcal{V}_i$ $\cup$
  $\{s_i\}$ and $s_i$ is the root. All nodes in $_r\mathcal{V}^i$ are
  directly connected to $s_i$ in $T_i$, $i.e.,$ $E^{T_i}$ $=$
  $\{(s_i, v)~|~v \in~_r\mathcal{V}_i\}$. 
\end{definition}

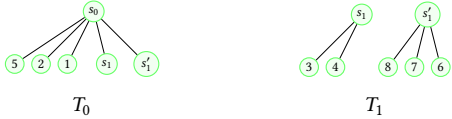
\begin{figure}[htb!]
  \begin{subfigure}{0.45\columnwidth} \centering
    \scalebox{.35}{
      \begin{tikzpicture}[roundnode/.style={circle, draw=green!60, fill=green!5, very thick, minimum size=6mm}]
      \node[roundnode] (r)  at (0.0, 1.0)  {$s_0$};

      \node[roundnode] (5)  at (-3, -1.0)  {5};
      
      \node[roundnode] (2)  at (-2.0, -1.0)  {2};
      
      \node[roundnode] (1)  at (-1.0, -1.0)  {1};
          
      

      \node[roundnode] (c2)  at (0.5, -1.0)  {$s_1$};
      \node[roundnode] (c1)  at (2.0, -1.0)  {$s_1'$};
      \draw[-] (r) to (5);
      \draw[-] (r) to (1);
      \draw[-] (r) to (2);
      \draw[-] (r) to (c1);
      \draw[-] (r) to (c2);
      
      
      
    \end{tikzpicture}
    }
    \caption*{$T_0$}
  \end{subfigure}
  \begin{subfigure}{0.45\columnwidth}\centering
    \scalebox{0.35}{
      \begin{tikzpicture}[roundnode/.style={circle, draw=green!60, fill=green!5, very thick, minimum size=6mm}]
          
      \node[roundnode] (3)  at (-2.5, -3.0)  {3};
      
      \node[roundnode] (4)  at (-1.5, -3.0)  {4};
      
      \node[roundnode] (c2)  at (-0.5, -1.0)  {$s_1$};
      \node[roundnode] (c1)  at (2.0, -1.0)  {$s_1'$};

      \draw[-] (c2) to (3);
      \draw[-] (c2) to (4);
      
      \node[roundnode] (8)  at (0.5, -3.0)  {8};
      \node[roundnode] (7)  at (1.5, -3.0)  {7};
      \node[roundnode] (6)  at (2.5, -3.0)  {6};
      
      \draw[-] (c1) to (8);
      \draw[-] (c1) to (7);
      \draw[-] (c1) to (6);
      
    \end{tikzpicture}
    }
    \caption*{$T_1$}
  \end{subfigure}
  \caption{Structural trees for recursive spanning trees
  of Figure~\ref{fig:exam_rst}.}
  \label{fig:exam_st}
\end{figure}
\begin{example}
  Figure~\ref{fig:exam_st} shows the structural tree $T_0$ 
  for the level-0 spanning tree and the structural tree $T_1$ 
  for the level-1 spanning trees in Figure~\ref{fig:exam_rst}.
\end{example}
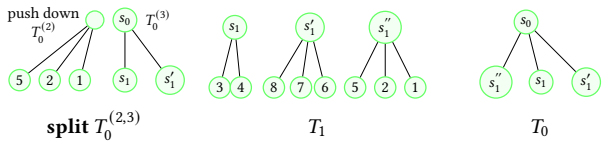
\begin{figure}[htb!]
  \begin{subfigure}{0.30\columnwidth} \centering
    \scalebox{.40}{
      \begin{tikzpicture}[roundnode/.style={circle, draw=green!60, fill=green!5, very thick, minimum size=6mm}]
      \node[roundnode] (r)  at (0.5, 1.0)  {$s_0$};
      \node[roundnode] (r1)  at (-0.5, 1.0)  {};

      \node  at (-2.2, 1.2)  {push down};
      \node  at (-2.2, 0.6)  {$T_0^{(2)}$};
      \node  at (1.6, 1.0)  {$T_0^{(3)}$};

      \node[roundnode] (5)  at (-3, -1.0)  {5};
      \node[roundnode] (2)  at (-2.0, -1.0)  {2};
      
      \node[roundnode] (1)  at (-1.0, -1.0)  {1};

      \node[roundnode] (c2)  at (0.5, -1.0)  {$s_1$};
      \node[roundnode] (c1)  at (2.0, -1.0)  {$s_1'$};
      \draw[-] (r1) to (5);
      \draw[-] (r1) to (1);
      \draw[-] (r1) to (2);
      \draw[-] (r) to (c1);
      \draw[-] (r) to (c2);      
    \end{tikzpicture}
    }
    \caption*{split $T_0^{(2, 3)}$}
  \end{subfigure}
  \begin{subfigure}{0.38\columnwidth}\centering
    \scalebox{.40}{
      \begin{tikzpicture}[roundnode/.style={circle, draw=green!60, fill=green!5, very thick, minimum size=6mm}]
          
      \node[roundnode] (3)  at (-1.5, -1.0)  {3};
      
      \node[roundnode] (4)  at (-0.8, -1.0)  {4};
      
      \node[roundnode] (c2)  at (-1.0, 1.0)  {$s_1$};
      \node[roundnode] (c1)  at (1.5, 1.0)  {$s_1'$};

      \draw[-] (c2) to (3);
      \draw[-] (c2) to (4);
      
      \node[roundnode] (8)  at (0.3, -1.0)  {8};
      \node[roundnode] (7)  at (1.2, -1.0)  {7};
      \node[roundnode] (6)  at (2.0, -1.0)  {6};
      
      \draw[-] (c1) to (8);
      \draw[-] (c1) to (7);
      \draw[-] (c1) to (6);
   
      \node[roundnode] (r1)  at (4.0, 1.0)  {$s_1^{''}$};
      \node[roundnode] (5)  at (3.0, -1.0)  {5};
      \node[roundnode] (2)  at (4.0, -1.0)  {2};
      
      \node[roundnode] (1)  at (5.0, -1.0)  {1};
      \draw[-] (r1) to (5);
      \draw[-] (r1) to (1);
      \draw[-] (r1) to (2);
    \end{tikzpicture}
    }
    \caption*{$T_1$}
  \end{subfigure}
  \begin{subfigure}{0.30\columnwidth}\centering
    \scalebox{.40}{
      \begin{tikzpicture}[roundnode/.style={circle, draw=green!60, fill=green!5, very thick, minimum size=6mm}]
      \node[roundnode] (r)  at (0.0, 1.0)  {$s_0$};
      \node[roundnode] (c3)  at (-1.0, -1.0)  {$s_1^{''}$};
      \node[roundnode] (c2)  at (0.5, -1.0)  {$s_1$};
      \node[roundnode] (c1)  at (2.0, -1.0)  {$s_1'$};
      \draw[-] (r) to (c3);
      \draw[-] (r) to (c1);
      \draw[-] (r) to (c2);

    \end{tikzpicture}
    }
    \caption*{$T_0$}
  \end{subfigure}
  \caption{Deleting the tree edge (2, 3) in the structural tree.}
  \label{fig:exam_st_de}
\end{figure}

Inserting a tree edge $(u, v)$ connects structural trees $T^{(u)}_0$
and $T^{(v)}_0$. $W.l.o.g$, assume that $T^{(u)}_0$ contains fewer
leaf nodes (super nodes are not leaf nodes).  Nodes of $T^{(u)}_0$
except the root directly connect to the root of $T^{(v)}_0$. The
performance of the insertion is determined by the number of nodes in
$T^{(u)}_0$.  The worst-case performance for inserting a tree edge is
$\Theta(n/2)$ if $T^{(u)}_0$ contains $n/2$ nodes.

When a level-$i$ tree edge $(u, v)$ is deleted, the level-$i$
structural tree $T^{(u, v)}_i$ that contains $u$ and $v$ splits into
two level-$i$ structural trees: $T^{(u)}_i$ and $T^{(v)}_i$.  
Let $T^{(u)}_i$ be the tree with
fewer leaf nodes. We traverse the leaf nodes in $T^{(u)}_i$ to find a
non-tree edge that reconnects $T^{(u)}_i$ and $T^{(v)}_i$.  When the
traversal stops, the levels of all traversed tree edges in $T^{(u)}_i$
and visited level-$i$ non-tree edges is increased by 1. Thus, the
level-$i$ structural tree $T^{(u)}_i$ becomes a level-$(i+1)$
structural tree.  The heuristic of pushing down the smaller tree
guarantees a $\log n$ height of a structural tree \cite{HDT} as there
are at most $\log n$ levels. A level-1 super node contains at most
$n/2$ leaf nodes that are all from the small tree. A level-2 super
node contains at most $n/4$ leaf nodes and so on.  
 The lower bound for the amortized costs for deleting tree
edges in structural trees is $O((\log n)^2)$ 
 \cite{wulff2013faster}. The worst-case performance of deleting a
tree edge is $O(|E|)$ when all tree edges and non-tree edges have to
be traversed and pushed to the next level.

\begin{example}
  Figure~\ref{fig:exam_st_de} illustrates the split of $T_0^{(2, 3)}$
  when the level-0 tree edge (2, 3) is deleted.  The structural tree
  $T_0^{(2)}$ containing node 2 has three leaf nodes while $T_0^{(3)}$
  containing node 3 has five leaf nodes, two leaf nodes in $s_1$ and
  three leaf nodes in $s_1'$. Nodes in $T_0^{(2)}$ are pushed down,
  and $T_0^{(2)}$ becomes $T_1^{(2)}$. We traverse leaf nodes of
  $T_1^{(2)}$ and find that level-0 non-tree edges (1, 3) and (3, 6)
  are replacement edges. $T_0^{(2)}$ with root $s_1{''}$ connects to
  the root of $T_0^{(3)}$.
\end{example}

All leaf nodes of the smaller tree have to be 
traversed when searching for a replacement edge,
even though they may not be connected via non-tree edges. 
In the worst case, if none of the leaf nodes have non-tree 
neighbors, we traverse all of them unnecessarily.
The local tree, which we look at next, transforms the 
structural tree to a 
height bounded binary tree with auxiliary bit arrays
to prune leaf nodes without non-tree neighbors. 

\subsubsection{Local Tree}
\label{sec:lt}
We first show how the local 
tree\footnote{In the original paper~\cite{thorup2000near}, 
local tree was called local search tree. But there is actually 
no search in the local tree, which is also confirmed by a 
recent work~\cite{huang2023fully}.} achieves and maintains
a bounded height and then 
show that the runtime of finding a non-tree edge 
in the local tree is bounded by the height.

\textbf{Bounded tree height.} A level-$i$ local tree 
achieves a bounded height by maintaining a bounded number 
($\leq \log n$)  of \emph{rank trees}, $i.e.,$
binary trees with bounded heights. 
Given a node $x$ of a local tree, 
we use $x.nl$ to denote the number of leaf nodes in the tree 
rooted at $x$. The \emph{rank} of a node $x$, 
$rank(x)$, is equal to 
$\lfloor \log_2 (x.nl) \rfloor$~\cite{thorup2000near}. 
Two nodes with the same rank are paired up and form a new
rank tree. Assume that nodes $x$ and $y$ have the same 
rank. $W.l.o.g.,$ $x.nl$ $\leq$ $y.nl$. 
We pair up $x$ and $y$, 
making  $x$ and $y$ the left child and the right 
child, respectively, of a new node $par$. 
The binary tree rooted at $par$ is called a rank tree
and the node $par$ is called a rank 
root (the root of a rank tree). 
If $par$ is paired with another node, $par$ is not a rank
root any more. One node of a local tree can be paired 
up only once. When no further pair operations can be done, 
at most $\log n$ rank trees with unique ranks 
remain. The height of each rank tree is $O(\log n)$ . 
Finally, the local trees connects all rank roots with a path
such that rank roots with larger ranks are closer to the root
node of the local tree. 

\textbf{Leverage bitmaps to search for a non-tree edge.}
Each node $x$ of the local tree has an array of $l_{max}$ 
bits, called bitmap, in which the $i$-th bit tells us if 
there exists a leaf node of the subtree rooted at $x$ that 
has a level-$i$ non-tree neighbor. 
The search for a level-$i$ non-tree edge traverses from the 
root to leaf nodes and is done by checking if 
bitmap[$i$] = 1 for the nodes on the path. This prunes the 
nodes that do not have non-tree neighbors.
The runtime of finding a non-tree edge is bounded by the
tree height. 
After we find a non-tree edge, we check
if the non-tree edge is a replacement edge. 

\begin{definition}[Level-i Local Tree]
Given a level-$i$ structural tree $T_i$ $=$ 
$(V^{T^i}, E^{T^i})$ with $s_i$ being the root,
a level-$i$ local tree for $T_i$ is $L_i$ $=$ 
$(V^{L^i}, E^{L^i})$ in which $V^{T^i}$ $\subseteq$ 
$V^{L^i}$ and $s_i$ is the root node. 
The distance between a node $x$ and the root $s_i$ 
in $L^i$ is $O(1 + \log(s_i.nl/x.nl))$ (for details, 
see~\cite{thorup2000near}).
\end{definition}

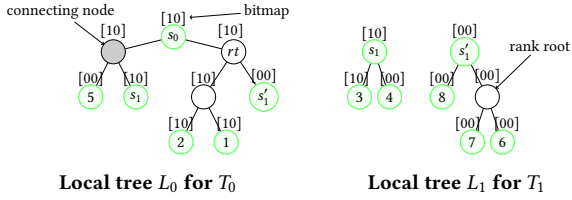
\begin{figure}[htb!]
  \begin{subfigure}{0.50\columnwidth}\centering
    \scalebox{0.40}{
    \begin{tikzpicture}

      \node  at (-1.5, 0.6)  {[00]};
      \node[roundnode] (5)  at (-1.5, 0.0)  {$5$};
      
      \node  at (0.0, 0.6)  {[10]};
      \node[roundnode] (s1)  at (0.0, 0.0)  {$s_1$};
 
      \node  at (-0.75, 2.2)  {[10]};
      \node[circle, fill=gray!40, draw, minimum size=8mm] (p1)  at (-0.75, 1.5)  {};
      \node (connecting) at (-2.5, 2.85) {connecting node};
      \draw[->] (connecting) to (p1);
      
      \node  at (2.25, 0.6)  {[10]};
      \node[circle, draw, minimum size=8mm] (l1)  at (2.25, 0.0)  {};
      
      \node  at (1.5, -0.9)  {[10]};
      \node[roundnode] (2)  at (1.5, -1.5)  {$2$};
       
      \node  at (3.1, -0.9)  {[10]}; 
      \node[roundnode] (1)  at (3.0, -1.5)  {$1$};
      
      \draw[-] (p1) to (5);
      \draw[-] (l1) to (1);
      \draw[-] (l1) to (2);
      
      \node  at (3.15, 2.2)  {[10]};
      \node[circle, draw, minimum size=8mm] (l2)  at (3.25, 1.5)  {$rt$};
      
      \draw[-] (p1) to (s1);
      
      \node () at (4.25, 0.7) {[00]};
      \node[roundnode] (s1p)  at (4.25, 0.0)  {$s_1'$};
      
      \draw[-] (l2) to (l1);
      \draw[-] (l2) to (s1p);  
      
      \node (bt_label) at (4.35, 2.85) {bitmap};
      \node (bt_value) at (1.25, 2.6)  {[10]};
      \node[roundnode] (p2)  at (1.25, 2.0)  {$s_0$};
      \draw[->] (bt_label) to (bt_value);
      
      \draw[-] (p2) to (p1);
      \draw[-] (p2) to (l2);


    \end{tikzpicture}
  }
  \caption*{Local tree $L_0$ for $T_0$}
  \end{subfigure}
  \begin{subfigure}[b]{0.45\columnwidth}\centering
    \scalebox{0.40}{
    \begin{tikzpicture}[roundnode/.style={circle, draw=green!60, 
    fill=green!5, very thick, minimum size=8mm}]
      \node ()  at (1.5, 0.6)  {[10]};
      \node[roundnode] (s1)  at (1.5, 0.0)  {$s_1$};

      \node ()  at (0.9, -0.9)  {[10]};
      \node[roundnode] (3)  at (1.0, -1.5)  {$3$};

      \node ()  at (2.1, -0.9)  {[00]};
      \node[roundnode] (4)  at (2.0, -1.5)  {$4$};
      
      \draw[-] (s1) to (3);
      \draw[-] (s1) to (4);

      \node[] ()  at (4.5, 0.7)  {[00]};
      \node[roundnode] (s1p)  at (4.5, 0.0)  {$s_1'$};

      \node[] ()  at (3.75, -0.9)  {[00]};
      \node[roundnode] (8)  at (3.75, -1.5)  {$8$};

      \node[] ()  at (5.25, -0.8)  {[00]};
      \node[circle, draw, minimum size=8mm] (rs1)  at (5.25, -1.5)  {};
      \node (rt_label)  at (7.0, 0.2)  {rank root};
      \draw[->] (rt_label) to (rs1);

      \node[] ()  at (4.60, -2.4)  {[00]};
      \node[roundnode] (7)  at (4.75, -3.0)  {$7$};

      \node[] ()  at (5.9, -2.4)  {[00]};
      \node[roundnode] (6)  at (5.75, -3.0)  {$6$};
      
      \draw[-] (s1p) to (8);
      \draw[-] (s1p) to (rs1);
      \draw[-] (rs1) to (7);
      \draw[-] (rs1) to (6);
    \end{tikzpicture}
    }
    \caption*{Local tree $L_1$ for $T_1$}
  \end{subfigure}
  \caption{Local trees for structural trees of 
    Figure~\ref{fig:exam_st}. Non-green nodes are internal nodes in
    the local tree.}
  \label{fig:exam_lt}
\end{figure}

\begin{example}
In Figure~\ref{fig:exam_lt}, rank roots with unique 
ranks are 5, $s_1$, $rt$ in which $rt$ is directly linked 
to the root $s_0$ while $5$ and $s_1$ are linked to 
$s_0$ via a connecting node. The $nl$ values for nodes 
5, $s_1$ and $rt$ are 1, 2 and 5, respectively. The 
\emph{rank} values for nodes 5, $s_1$ and $rt$ are 0 
(=$\lfloor 1 \rfloor$), 1 (=$\lfloor 2 \rfloor$), and 
2 (=$\lfloor s_1'.nl \rfloor$ + 1 = 1 + 1), respectively.
\end{example}



When a level-$i$ tree edge $(u, v)$ is deleted, 
the local tree 
changes $st_i^u$ to $st_{i+1}^u$, pushing
down all tree edges and nodes in $st_i^u$. All nodes
of $st_i^u$ are firstly removed from the local tree 
$L_i^{u, v}$ of $st_i^{u, v}$, and 
then form a level-$(i+1)$ local tree $L_{i+1}^u$ for
$st_{i+1}^u$. 
We search $L_{i+1}^u$ for a non-tree edge
that reconnects $L_{i+1}^u$ and $L_{i}^v$. The search
starts at the root of $L_{i+1}^u$ and goes to the
left or right subtree whose bitmap[i] is 1.  
If there exists a level-$(i+1)$ non-tree 
edge that reconnects $st_i^u$ and $st_i^v$, $L_{i+1}^u$ 
is inserted into $L_{i}^v$, the level-$i$ local tree that 
contains $v$ for $st_i^v$, as a level-$(i+1)$ node. 
All level-$i$ non-tree edges visited during the search
for a replacement, except the replacement edge itself,
are pushed down to level $i+1$, and 
hence we update the bitmap accordingly.

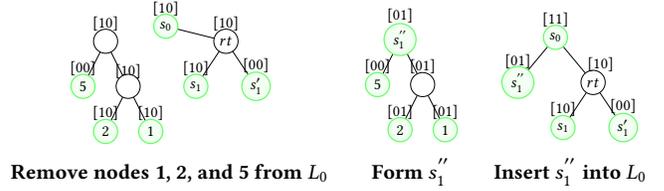
\begin{figure}[htb!]
  \begin{subfigure}{0.50\columnwidth}\centering
    \scalebox{0.40}{
    \begin{tikzpicture}[roundnode/.style={circle, draw=green!60, 
      fill=green!5, very thick, minimum size=8mm}]

      \node  at (-1.5, 0.6)  {[00]};
      \node[roundnode] (5)  at (-1.5, 0.0)  {$5$};
      
      \node  at (0.0, 0.5)  {[10]};
      \node[roundnode] (s1)  at (2.25, 0.0)  {$s_1$};
      
      \node  at (-0.75, 2.1)  {[10]};
      \node[circle, draw, minimum size=8mm] (p1)  at (-0.75, 1.5)  {};
      
      \node  at (2.25, 0.6)  {[10]};
      \node[circle, draw, minimum size=8mm] (l1)  at (0.0, 0.0)  {};
      
      \node  at (-0.75, -0.9)  {[10]};
      \node[roundnode] (2)  at (-0.75, -1.5)  {$2$};
       
      \node  at (0.75, -0.9)  {[10]}; 
      \node[roundnode] (1)  at (0.75, -1.5)  {$1$};
      
      \draw[-] (p1) to (5);
      \draw[-] (l1) to (1);
      \draw[-] (l1) to (2);
      
      \node  at (3.25, 2.1)  {[10]};
      \node[circle, draw, minimum size=8mm] (l2)  at (3.25, 1.5)  {$rt$};
      
      \draw[-] (l2) to (s1);
      
      \node () at (4.25, 0.7) {[00]};
      \node[roundnode] (s1p)  at (4.25, 0.0)  {$s_1'$};
      
      \draw[-] (p1) to (l1);
      \draw[-] (l2) to (s1p);  
      
      \node  at (1.25, 2.6)  {[10]};
      \node[roundnode] (s0)  at (1.25, 2.0)  {$s_0$};
      
      \draw[-] (s0) to (l2);


    \end{tikzpicture}
  }
  \caption*{Remove nodes 1, 2, and 5 from $L_0$}
  \end{subfigure}
  \begin{subfigure}[b]{0.24\columnwidth}\centering
    \scalebox{0.40}{
    \begin{tikzpicture}[roundnode/.style={circle, draw=green!60, 
      fill=green!5, very thick, minimum size=8mm}]

      \node  at (-1.5, 0.6)  {[00]};
      \node[roundnode] (5)  at (-1.5, 0.0)  {$5$};

      \node  at (-0.75, 2.3)  {[01]};
      \node[roundnode] (p1)  at (-0.75, 1.5)  {$s_1^{''}$};
      
      \node  at (0.0, 0.6)  {[01]};
      \node[circle, draw, minimum size=8mm] (l1)  at (0.0, 0.0)  {};
      
      \node  at (-0.75, -0.9)  {[01]};
      \node[roundnode] (2)  at (-0.75, -1.5)  {$2$};
       
      \node  at (0.75, -0.9)  {[01]}; 
      \node[roundnode] (1)  at (0.75, -1.5)  {$1$};
      
      \draw[-] (p1) to (5);
      \draw[-] (l1) to (1);
      \draw[-] (l1) to (2);
      
      \draw[-] (p1) to (l1);
      
    \end{tikzpicture}
    }
    \caption*{Form $s_1^{''}$}
  \end{subfigure}
  \begin{subfigure}[b]{0.24\columnwidth}\centering
    \scalebox{0.40}{
    \begin{tikzpicture}[roundnode/.style={circle, draw=green!60, 
    fill=green!5, very thick, minimum size=8mm}]

      \node  at (2.0, 3.6)  {[11]};
      \node[roundnode] (s0)  at (2.0, 3.0)  {$s_0$};

      \node  at (3.5, 2.1)  {[10]};
      \node[circle, draw, minimum size=8mm] (rt)  at (3.25, 1.5)  {$rt$};

      \node at (0.75, 2.2) {[01]}; 
      \node[roundnode] (s1t)  at (0.75, 1.5)  {$s_1^{''}$};
      
      \node  at (2.25, 0.6)  {[10]};
      \node[roundnode] (s1)  at (2.25, 0.0)  {$s_1$};
      
      \node () at (4.25, 0.7) {[00]};
      \node[roundnode] (s1p)  at (4.25, 0.0)  {$s_1'$};

      \draw[-] (s0) to (s1t);
      \draw[-] (s0) to (rt);
      \draw[-] (rt) to (s1);
      \draw[-] (rt) to (s1p);
    \end{tikzpicture}
    }
    \caption*{Insert $s_1^{''}$ into $L_0$}
  \end{subfigure}
  \caption{Deleting tree edge (2, 3) in the local tree.}
  \label{fig:exam_lt_de}
\end{figure}

\begin{example}
  Figure~\ref{fig:exam_lt_de} shows that tree edge $(2, 3)$ 
  is deleted, nodes 1, 2, 5 are  removed from the the local 
  tree and combined to a level-$(i+1)$ local tree rooted at 
  $s_1^{''}$. The levels of tree edges (2, 5) and (1, 5) 
  are increased
  by 1. The search for a non-tree edge starts at the 
  root node $s_1^{''}$ and we navigate to left or right
  subtrees whose root nodes have the $0$-th bit of the bitmap 
  set to 1. We find two level-0 non-tree edges (1, 2) and 
  (1, 3), of which non-tree edge 
  (1, 3) is a replacement edge.  
  Edge (1, 3) becomes a level-0 tree edge and $s_1^{''}$ is 
  inserted as a level-1 node in $L_0$. 
  We  update the bitmaps for the tree rooted at
  $s_1^{''}$ as the level of 
  non-tree edge (1, 2) increases by 1.
  \label{exam:lt_de}
\end{example}

To further improve the amortized costs for deleting tree 
edges, Thorup proposed the Lazy Local Tree~\cite{thorup2000near} 
that lazily updates the binary tree by additionally 
maintaining a small subtree, called buffer tree.

\subsubsection{Lazy Local Tree}
\label{sec:lzt}

The Lazy Local Tree~\cite{thorup2000near} maintains nodes with small
$nl$ values (number of leaf nodes) and nodes with large $nl$ values separately.  A lazy local
tree maintains two local trees in left and right subtrees 
of the root node, respectively.  The subtree in the left 
branch, called \emph{buffer tree}, contains nodes whose $nl$ values 
are below a threshold $\beta$. The subtree in 
the right branch contains nodes with
$nl$ values $\geq$ $\beta$ or a group of buffer nodes (described in
the following) with a total $nl$ value $\geq$ $\beta$. We call a node
of the buffer tree a \emph{buffer node} and call the right branch of
the lazy local tree \emph{lazy branch}.
A node whose $nl$ value is below $\beta$ can only be
inserted into the buffer tree.  When the $nl$ value of the buffer tree
is equal to or larger than $\beta$, the buffer tree becomes a
\emph{bottom tree}.  The buffer tree (now a bottom tree) is moved 
to the lazy branch, $i.e.,$ inserting the root node of the bottom 
tree as a rank root into the lazy branch (a local tree). 
The left branch of the lazy local tree
is empty.  Inserting a node with
$nl$ value $\geq$ $\beta$ into a lazy local tree is the same as
inserting a node in a local tree. 

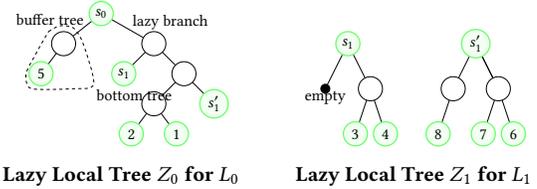
\begin{figure}[htb!]
  \begin{subfigure}[b]{0.5\columnwidth}\centering
    \scalebox{0.4}{
    \begin{tikzpicture}[roundnode/.style={circle, draw=green!60, 
      fill=green!5, very thick, minimum size=8mm}]
      \node[roundnode] (s0)  at (0.5, 3.0)  {$s_0$};
      \node[roundnode] (s1)  at (1.25, 1.0)  {$s_1$};
      \node[circle, draw, minimum size=8mm] (p1)  at (-0.75, 2.0)  {};
      \node[circle, draw, minimum size=8mm] (p2)  at (2.25, 2.0)  {};
      \node[circle, draw, minimum size=8mm] (l1)  at (2.25, 0.0)  {};
      \node[roundnode] (5)  at (-1.5, 1.0)  {$5$};
      \node[roundnode] (2)  at (1.5, -1.0)  {$2$};
      \node[roundnode] (1)  at (3.0, -1.0)  {$1$};
      \node[circle, draw, minimum size=8mm] (l2)  at (3.25, 1.0)  {};
      \node[roundnode] (s1p)  at (4.25, 0.0)  {$s_1'$};
      
      \draw [dashed] plot [smooth] coordinates {(-1.25, 2.4) (-1.0, 2.55) 
      (-0.75, 2.6) (-0.4, 2.5) (-0.1, 2.0) (0.25, 0.65) (-0.5, 0.5) 
      (-1.8, 0.5) (-2.0, 0.85) (-1.9, 1.25) (-1.25, 2.4)} ;
      \node at (-1.2, 2.8) {buffer tree};

      \node at (2.8, 2.7) {lazy branch};
      \node at (1.6, 0.3) {bottom tree};
      \draw[-] (s0) to (p1);
      \draw[-] (s0) to (p2);
      \draw[-] (p2) to (s1);
      \draw[-] (p2) to (l2);
      \draw[-] (p1) to (5);
      \draw[-] (l1) to (1);
      \draw[-] (l1) to (2);
      
      \draw[-] (l2) to (l1);
      \draw[-] (l2) to (s1p);  
    \end{tikzpicture}
    }
    \caption*{Lazy Local Tree $Z_0$ for $L_0$}
  \end{subfigure}
  \begin{subfigure}[b]{0.4\columnwidth}\centering
    \scalebox{0.4}{
      \begin{tikzpicture}[roundnode/.style={circle, draw=green!60, 
      fill=green!5, very thick, minimum size=8mm}]
        \node[roundnode] (s1)  at (1.25, 0.0)  {$s_1$};
        \node at (0.5, -1.8) {empty};
        \node[draw, fill, circle] (sl1)  at (0.5, -1.5)  {{\small }};
        \node[circle, draw, minimum size=8mm] (rs1)  at (2.0, -1.5)  {};    
        \draw[-] (s1) to (sl1);
        \draw[-] (s1) to (rs1);

        \node[roundnode] (3)  at (1.5, -3.0)  {$3$};
        \node[roundnode] (4)  at (2.5, -3.0)  {$4$};
        
        \draw[-] (rs1) to (3);
        \draw[-] (rs1) to (4);

        \node[roundnode] (s1p)  at (5.5, 0.0)  {$s_1'$};
        \node[circle, draw, minimum size=8mm] (ls1p)  at (4.75, -1.5)  {};
        \node[roundnode] (8)  at (4.25, -3.0)  {$8$};

        \node[circle, draw, minimum size=8mm] (rs1p)  at (6.25, -1.5)  {};

        \node[roundnode] (7)  at (5.75, -3.0)  {$7$};
        \node[roundnode] (6)  at (6.75, -3.0)  {$6$};
        
        \draw[-] (s1p) to (rs1p);
        \draw[-] (s1p) to (ls1p);
        \draw[-] (s1p) to (rs1p);
        \draw[-] (ls1p) to (8);
        \draw[-] (rs1p) to (7);
        \draw[-] (rs1p) to (6);
      \end{tikzpicture}
    }
    \caption*{Lazy Local Tree $Z_1$ for $L_1$}
  \end{subfigure}
  \caption{Lazy Local trees with $\beta$ $=$ 2 
  for local trees of Figure~\ref{fig:exam_lt}. Bitmaps are not shown
  for simplicity.
  }
  \label{fig:exam_lzt}
\end{figure}

\begin{example}
  Figure~\ref{fig:exam_lzt} shows lazy local trees for the local tree
  in Figure~\ref{fig:exam_lt}. In $Z_0$, $s_1$ can not be placed in
  the buffer tree as its $nl$ value is equal to $\beta$. After nodes 1
  and 2 are inserted into the buffer tree, the buffer tree become a
  bottom tree which is moved to the lazy branch.
\end{example}

When a level-$i$ tree edge $(u, v)$ is deleted, the lazy
local tree changes $st_i^u$ to $st_{i+1}^u$. All nodes of $st_i^u$ are
removed from the lazy local tree $Z_i^{u, v}$ for $st_i^{u, v}$. There
are three types of nodes in $st_i^u$: (1) buffer nodes (2) nodes in a
bottom tree 
(3) intermediate nodes with $nl$ $\geq$ $\beta$. Removing
(1) and (3) from $Z_i^{u, v}$ works is the same as removing nodes from
a local tree.  Removing nodes from a bottom tree reduces the its $nl$
and the bottom tree is removed from the lazy branch if its $nl$ value
is below $\beta$. Nodes of the removed bottom tree are inserted into
the buffer tree. 
After all nodes of
$st_i^u$ are removed from $Z_i^{u, v}$, they form a new level-$(i+1)$
lazy local tree lazy local tree $Z_{i+1}^u$ for $st_{i+1}^u$.  If
there exists a level-$(i+1)$ non-tree edge that reconnects $st_i^u$
and $st_i^v$, $Z_{i+1}^u$ is inserted into $Z_i^v$.

\begin{figure}[htb!]
  \begin{subfigure}{0.50\columnwidth}\centering
    \scalebox{0.40}{
    \begin{tikzpicture}[roundnode/.style={circle, draw=green!60, 
      fill=green!5, very thick, minimum size=8mm}]
      \node[roundnode] (s0)  at (2.25, 1.5)  {$s_0$};
      \node[circle, draw, minimum size=8mm] (rt)  at (3.25, 0.5)  {$rt$};
      \node[draw, fill, circle] (e)  at (1.25, 0.5)  {};
      \node[roundnode] (s1p)  at (4.25, -1.0)  {$s_1'$};
      \node[roundnode] (s1)  at (2.25, -1.0)  {$s_1$};
      \node[roundnode] (5)  at (-1.5, 0.0)  {$5$};

      \node[roundnode] (2)  at (-0.75, -1.0)  {$2$};
      \node[roundnode] (1)  at (0.25, -1.0)  {$1$};

      \draw[-] (rt) to (s1);

      \draw[-] (rt) to (s1p);  
      \draw[-] (s0) to (e);
      \draw[-] (s0) to (rt);
    \end{tikzpicture}
    }
    \caption*{Remove nodes 1, 2, and 5 from $Z_0$}
  \end{subfigure}
  \begin{subfigure}[b]{0.24\columnwidth}\centering
    \scalebox{0.40}{
      \begin{tikzpicture}[roundnode/.style={circle, draw=green!60, 
        fill=green!5, very thick, minimum size=8mm}]
        \node[roundnode] (5)  at (-1.75, -1.5)  {$5$};
        \node[roundnode] (r)  at (-0.75, 1.5)  {$s_1^{''}$};
        \node[circle, draw, minimum size=8mm] (l1)  at (-1.25, 0.0)  {};
        \node[circle, draw, minimum size=8mm] (l2)  at (0.0, 0.0)  {};
        \node[roundnode] (2)  at (-0.5, -1.5)  {$2$};
        \node[roundnode] (1)  at (0.75, -1.5)  {$1$};
        \draw[-] (l1) to (5);
        \draw[-] (l2) to (1);
        \draw[-] (l2) to (2);
        \draw[-] (r) to (l2);
        \draw[-] (r) to (l1);
      \end{tikzpicture}
    }
    \caption*{Form $s_1^{''}$}
  \end{subfigure}
  \begin{subfigure}[b]{0.24\columnwidth}\centering
    \scalebox{0.40}{
    \begin{tikzpicture}[roundnode/.style={circle, draw=green!60, 
        fill=green!5, very thick, minimum size=8mm}]
      \node[roundnode] (s0)  at (1.75, 3.0)  {$s_0$};
      \node[circle, draw, fill] (l1)  at (1.0, 2.0)  {};
      \node[circle, draw, minimum size=8mm] (l2)  at (2.5, 2.0)  {};
      \node[circle, draw, minimum size=8mm] (rt)  at (3.25, 1.0)  {$rt$};
      \node[roundnode] (s1t)  at (1.75, 1.0)  {$s_1^{''}$};
      \node[roundnode] (s1)  at (2.5, 0.0)  {$s_1$};
      \node[roundnode] (s1p)  at (4.0, 0.0)  {$s_1'$};

      \draw[-] (s0) to (l2);
      \draw[-] (s0) to (l1);
      \draw[-] (l2) to (rt);
      \draw[-] (l2) to (s1t);
      \draw[-] (rt) to (s1);
      \draw[-] (rt) to (s1p);
    \end{tikzpicture}
    }
    \caption*{Insert $s_1^{''}$ into $Z_0$}
  \end{subfigure}
  \caption{Operations for deleting the tree edge (2, 3) in the 
  lazy local tree $Z_0$.}
  \label{fig:exam_lzt_de}
\end{figure}
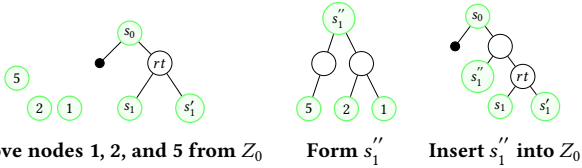

\begin{example} \label{exam:lzt_de}%
  Figure~\ref{fig:exam_lzt_de} shows the operations for deleting the
  tree edge (2, 3) in $Z_0$. The bottom tree that contains nodes 1 and
  2 is removed from the lazy branch as removing either node 1 or node
  2 makes its $nl$ value less than $\beta$. Thus, the remained node in
  the bottom tree is removed from the lazy branch and added to the
  buffer tree.  Nodes 1, 2 and 5 form a level-($i+1$) lazy local tree
  rooted at $s_1^{''}$, which is inserted into the lazy branch of
  $Z_0$.
\end{example}

\section{Implementations}
\label{sec:imp}

Table~\ref{tab:imp} shows that many data structures discussed in
Section~\ref{sec:maintaining_ds} were not implemented before or use
outdated third-party libraries.  An link-cut tree cannot be
used for the dynamic connectivity problem directly as it does not
  support the search of a replacement edge and the change of root node
  of a path. Thus, we extend existing
  implementations\footnote{\url{https://usaco.guide/adv/link-cut-tree?lang=cpp}} with this functionality.
To compare the data structures, we have implemented all of them in Python.
In this section, we describe important details for implementing the
data structures and discuss how we measure memory footprints.

\begin{table}[htb!]\centering
  \begin{tabular} {|c|c|c|}
    \hline
    \textbf{Data structure} & \textbf{Year} & \textbf{Implementation}\\
    \hline
    Link-cut tree (LCT) & 1981 & \cmark \\
    \hline
    ET-tree (HKS) & 1997 & (\cmark)\\
    \hline
    ET-tree (HK)  & 1995 & (\cmark)\\
    \hline
    ET-tree (HDT) & 1998 & (\cmark) \\
    \hline
    Structural Tree (ST) & 2000 & \xmark\\
    \hline
    Structural Tree variant (STV) & 2013 & \xmark\\
    \hline
    Local Tree (LT) & 2000 & \xmark\\
    \hline
    Local Tree variant (LTV) & 2013 & \xmark\\
    \hline
    Lazy Local Tree (LzT) & 2000 & \xmark\\
    \hline
    D-tree & 2022 & \cmark\\
    \hline
  \end{tabular}
  \caption{Status for major data structures. The symbol "(\cmark)"
    means that the implementation is based on an outdated third-party
    library that is not currently maintained. }
  \label{tab:imp}
\end{table}


\subsection{D-tree}
\label{sec:imp_dtree}

Each node $u$ in the D-tree has the following attributes:
\begin{itemize}
\item $key$: unique integer identifying $u$
\item $children$: set of pointers to the children of $u$
\item $parent$: pointer to parent of $u$
\item $size$: number of nodes in subtree rooted at $u$
\item $nte$: set of non-tree neighbors of $u$
\end{itemize}
A D-tree maintains a dictionary to look up tree nodes.  Thus, 
given a key $x$, the dictionary returns the tree node with $key$ 
equal to $x$. 
Implementation details are shown in ~\cite{dtree}.

\subsection{Link-cut Tree}
\label{sec:imp_lct}
Each node $u$ in the link-cut tree (LCT) has the following attributes:
\begin{itemize}
\item $key$: unique integer identifying $u$
\item $left$: pointer to left child of $u$ in the splay tree 
\item $right$: pointer to right child of $u$ in the splay tree 
\item $parent$: pointer to parent of $u$ in the splay tree 
\item $nte$: set of non-tree neighbors of $u$ in the spanning tree
\item $path\_parent$: pointer to parent of $u$ in the LCT
\end{itemize}
An LCT maintains a dictionary to look up tree nodes.

\subsection{HKS, HK and HDT}
\label{sec:imp_hks}
Existing implementations~\cite{Alberts1995, David_HKvariant,
  Iyer2002experimental} for HKS, HK and HDT use randomized search
trees~\cite{seidel1996randomized} to implement ET-trees
\cite{Alberts1995}.  A randomized search tree is a balanced binary
search tree where the search key of each node, the \emph{priority}, 
is a random integer. Tree nodes whose keys are equal to 
the first occurrences are active nodes.
For example, Figure~\ref{fig:exam_et} shows an ET-tree, 
in which nodes with keys equal to $1_1$, $2_1$, ..., $8_1$ 
are active nodes.  All level-$i$ non-tree neighbors of a vertex 
$u$ in the graph are stored in the active node with key $u_1$ 
in the ET-tree. HKS, HK and HDT use different  data structures 
to keep track of edge and node information. We start out with the 
attributes maintained by HK. Each node $u$ of HK has the 
following attributes:
\begin{itemize}
\item $key$: integer identifier of $u$ 
\item $left$: pointer to left child of $u$ 
\item $right$: pointer to right child of $u$ 
\item $parent$: pointer to parent of $u$ 
\item $priority$: priority of $u$ 
\item $active$: true for $u_1$ of key $u$ 
\item $nte$: set of non-tree neighbors of $u$ 
\item $weight$: number of non-tree edges in tree rooted at $u$ 
\item $act\_dict\_i$: maps $u$ to the level-$i$ tree node with key $u_1$ 
\item $rst$: a randomized search tree where nodes are non-tree neighbors of $u$ 
\end{itemize}

HK maintains the following attributes for keeping tracking of Euler tours
and levels of edges:
\begin{itemize}
  \item $tree\_edge\_pointer\_i$: maps a tree edge of a level-$i$ cumulative 
  spanning
  tree $(u, v)$ to four level-$(i+1)$ tree nodes with keys $u_x$, $v_y$, 
  $v_p$ and $u_q$ 
  \item $level\_of\_tree\_edges\_i$: set of level-$i$ tree edges
  \item $level\_of\_nontree\_edges\_i$: set of level-$i$ 
non-tree edges
\end{itemize}

\begin{example}
Consider the level-0 spanning tree $_c\mathcal{ST}_0$ and the 
ET-tree for $_c\mathcal{ST}_0$ in Figure \ref{fig:exam_et}. 
Table \ref{tab:exam_HK_attributes} shows 
examples for HK's 
$tree\_edge\_pointer\_0$ and $act\_dict\_0$.
\end{example}

\begin{table}[!htb]
  \begin{minipage}{.50\linewidth}
    \centering
    \begin{tabular}{ c | c  }
      tree edge & tree nodes \\
      \hline
      (7, 8) & $7_1$, $8_1$, $8_2$, $7_2$  \\ 
      (5, 2) & $5_1$, $2_1$, $2_2$, $5_2$  \\    
    \end{tabular}
    \caption*{$tree\_edge\_pointer\_0$}
  \end{minipage}%
  \hspace{1cm}
  \begin{minipage}{.33\linewidth}
    \centering
    \begin{tabular}{ c | c  }
      vertex & tree node \\
      \hline
      7 & $7_1$  \\ 
      8 & $8_1$\\    
    \end{tabular}
    \caption*{$act\_dict\_0$}
  \end{minipage} 
  \caption{Examples for $tree\_edge\_pointer\_0$ and 
            $act\_dict\_0$}
  \label{tab:exam_HK_attributes}
\end{table}

\newcommand{\myminus}{\rotatebox[origin=c]{90}{\textcolor{red}{\ding{121}}}}
\newcommand{\myplus}{\textcolor{green}{\ding{58}}}
\begin{figure}[htb!]
  \scalebox{0.7} {
    \begin{tikzpicture}
      \begin{scope}
        \node at (-1.8,0) [draw = green,name=HK,rectangle, minimum width=2.0cm,minimum height=1.0cm, anchor=south, font={\normalsize}] {$HK$};
        \node at (2.3,0) [draw,name=HDT,rectangle, minimum width=2.0cm,minimum height=1.0cm, anchor=south, font={\normalsize}] {$HDT$};
        \node[font={\normalsize}] ()  at (0.2, 0.8)  {\myplus size};
        \node[font={\normalsize}] ()  at (0.2, 0.3)  {\myminus weight};
        \node[font={\normalsize}] ()  at (0.2, 0.0)  {\myminus $rst$};
        
        \node at (-6.0,0) [draw,name=HKS,rectangle, minimum width=2.0cm,minimum height=1.0cm, anchor=south, font={\normalsize}] {$HKS$};
        \node[font={\normalsize}] ()  at (-4.0, 0.3)  {drop levels};

        \draw[line width=0.5mm, decoration = { markings, 
          mark=at position 1 with {\arrow[ultra thick]{>}}},
          postaction={decorate}] (HK) to (HKS);
        \draw[line width=0.5mm, decoration = { markings, 
          mark=at position 1 with {\arrow[ultra thick]{>}}},
          postaction={decorate}] (HK) to (HDT);
      \end{scope}
    \end{tikzpicture}
  }
  \caption{Attributes maintained by HK, HKS and HDT. Symbols \myplus~
    and \myminus~ represent adding and removing attributes,
    respectively.}
  \label{fig:attributes_ET}
\end{figure}
HKS simplifies HK by not maintaining attributes
related to levels. Attributes $act\_dict\_i$ and 
$tree\_edge\_pointer\_i$ are replaced by the following attributes:
\begin{itemize}
  \item $act\_dict$: maps $u$ to the tree node with key $u_1$ 
  \item $tree\_edge\_pointer$: maps a tree edge $(u, v)$ to 
  four tree nodes with keys $u_x$, $v_y$, $v_p$ and $u_q$
\end{itemize}
Attributes $level\_of\_tree\_edges\_i$ and 
$level\_of\_nontree\_edges\_i$ are not maintained by HKS. 
HDT maintain the $size$ attribute that is the number of nodes 
in tree rooted at $u$ instead of $weight$ attribute. 
HKS and HK store the non-tree neighbors of a tree node $u$ in 
a randomized search tree ~\cite{Alberts1995} to speed up the 
random selection of a non-tree neighbor for $u$. HDT does not use
randomized search trees for the non-tree neighbors, which saves 
space. 
Our experiments show that storing non-tree neighbors in 
randomized search trees yields high memory footprints (see details 
in Section \ref{sec:memory_footprint}).
Transformations among HK, 
HKS and HDT are shown in Figure~\ref{fig:attributes_ET}.

\subsection{Structural Trees and (Lazy) Local Trees}

There are two different versions of structural trees. The structural
tree (ST) proposed by Thorup~\cite{thorup2000near} maintains tree
edges and non-tree edges. This means that given a node $u$, level-$i$
neighbors of $u$ are classified into level-$i$ tree edge neighbors
$adj_i^t(u)$ and level-$i$ non-tree edge neighbors $adj_i^{nt}(u)$.
Wulff-Nilsen~\cite{wulff2013faster} proposed a variant of the
structural tree (STV) that does not differentiate tree and non-tree
edges.  Instead all level-$i$ neighbors of $u$ are stored in
$adj_i(u)$.  Since local trees are based on structural trees, there
are also two versions for local trees. We call the local tree proposed
by Thorup~\cite{thorup2000near} $LT$ and the variant of the local tree
proposed by Wulff-Nilsen~\cite{wulff2013faster} $LTV$.  The lazy local
tree (LzT) maintains the exact same attributes as LT as LzT is
essentially a local tree.  We start with the attributes maintained by
ST.  ST maintains the following attributes:
\begin{itemize}
  \item $key$: unique integer identifying $u$
  \item $children$: set of pointers to the children of $u$
  \item $parent$: pointer to parent of $u$ 
  \item $nl$: number of leaf nodes in tree rooted at $u$
  \item $level$: the level of $u$
  \item $adj\_t\_i$:  maps $u$ to level-$i$ tree neighbors 
  \item $adj\_nt\_i$: maps $u$ to level-$i$ non-tree neighbors
\end{itemize}

\begin{figure}[htb!]
  \scalebox{0.66} {
  \begin{tikzpicture}
    \begin{scope}
      \node at (0,3.5) [draw = orange,name=ST,rectangle, minimum width=2.5cm,minimum height=1.2cm, anchor=south, 
      font={\LARGE}] {$ST$};
      \node at (6.0,3.5) [draw,name=STV,rectangle, minimum width=2.5cm,minimum height=1.2cm, anchor=south, 
      font={\LARGE}] {$STV$};
      \node at (0,0) [draw,name=LT,rectangle, minimum width=2.5cm,minimum height=1.2cm, anchor=south, font={\LARGE}] {$LT$, $LzT$};
      \node at (6.0,0) [draw,name=LTV,rectangle, minimum width=2.5cm,minimum height=1.2cm, anchor=south, font={\LARGE}] {$LTV$};
      
      \draw[line width=0.5mm, decoration = {markings, 
        mark=at position 1 with {\arrow[ultra thick]{>}}},
        postaction={decorate}] (ST) to (STV);
      \node[font={\normalsize}] ()  at (3.0, 4.5)  {\myplus $adj_i$};
      \node[font={\normalsize}] ()  at (3.0, 3.8)  {\myminus $adj\_t\_i$};
      \node[font={\normalsize}] ()  at (3.0, 3.4)  {\myminus $adj\_nt\_i$};
      
      \draw[line width=0.5mm, decoration = {markings, 
        mark=at position 1 with {\arrow[ultra thick]{>}}},
        postaction={decorate}] (ST) to (LT);
      \node[font={\normalsize}] ()  at (0.85, 2.5)  {\myminus $children$};
      \node[font={\normalsize}] ()  at (-0.75, 3.25)  {\myplus $left$};
      \node[font={\normalsize}] ()  at (-0.75, 2.75)  {\myplus $right$};
      \node[font={\normalsize}] ()  at (-1.25, 2.25)  {\myplus $tree\_bitmap$};
      \node[font={\normalsize}] ()  at (-1.5, 1.75)  {\myplus $nontree\_bitmap$};
      
      \draw[line width=0.5mm, decoration = {markings, 
        mark=at position 1 with {\arrow[ultra thick]{>}}},
        postaction={decorate}] (STV) to (LTV);
      \node[font={\normalsize}] ()  at (5.0, 2.5)  {\myminus $children$};
      \node[font={\normalsize}] ()  at (6.75, 2.75)  {\myplus $left$};
      \node[font={\normalsize}] ()  at (6.75, 2.40)  {\myplus $right$};
      \node[font={\normalsize}] ()  at (6.75, 2.0)  {\myplus $bitmap$};
      
      \draw[line width=0.5mm, decoration = {markings, 
        mark=at position 1 with {\arrow[ultra thick]{>}}},
        postaction={decorate}] (LT) to (LTV);
      \node[font={\normalsize}] ()  at (3.0, 0.8)  {\myplus $$bitmap$$};
      \node[font={\normalsize}] ()  at (3.0, 0.2)  {\myminus $tree\_bitmap$};
      \node[font={\normalsize}] ()  at (3.0, -0.2)  {\myminus $nontree\_bitmap$};
    \end{scope}
  \end{tikzpicture}
  }
  \caption{Attributes maintained by ST, STV, LT, LTV and LzT. 
  Symbols \myplus~ and \myminus~ represent adding and removing 
  attributes, respectively.}
\end{figure}

STV maintains one attribute $adj\_i$ that maps a node $u$ to its
level-$i$ neighbors $adj_i(u)$ instead of two attributes $adj\_t\_i$
and $adj\_nt\_i$. Local trees are binary trees and hence have
attributes $left$ and $right$ that are pointers to left child and
right child of a node, respectively, instead of the $children$
attribute.  Moreover, nodes of local trees maintain bitmaps and node
types which are integers. LT maintains two node attributes for
bitmaps, $tree\_bitmap$ and $nontree\_bitmap$ that is an array of 64
bits for tree neighbors and non-tree neighbors of a node,
respectively.  LTV maintains one node attribute $bitmap$ that is an
array of 64 bits for neighbors of a node.

\subsection{Measuring Memory Footprints}

Some data structures have better space utilization than others due
to the techniques they are leveraging.  D-trees (spanning trees) use
less space as tree edges are stored in parent and children
attributes.
ET-trees by nature need more space as they have to keep track of
information for Euler tours.  
Data structures in Table~\ref{tab:imp} use integers, pointers, 
sets and maps (dictionaries) to store nodes attributes and edges. 
We determine the size of an object in bytes, 
to calculate the memory 
footprint for the data structures. 
The size for a set and dictionary is dynamic and determined 
by the number of elements.  
Low memory footprints are good in practice 
as they tend not to cause memory issues.

\section{Experimental Evaluation}
\label{sec:experiments}

We evaluate the data structures in Table~\ref{tab:imp} on a wide range
of synthetic and real-world graphs to work out the memory footprints
and performance for connectivity queries and update operations. Our
goal is determining the data structures that are feasible to be
deployed in graph database management systems (GDBMS).  Towards this
end, we conduct numerous experiments to identify the strengths and
weaknesses of each data structure in Table~\ref{tab:imp}.

\subsection{Setup}

All data structures and algorithms were implemented in Python 3. We
conduct experiments on a machine with 200 GB RAM and 80 GB swap
memory, running Debian 10.  We repeated all experiments for 5 times on
the same machine, which show very similar results.

\subsection{Datasets}
The diameters of real-world graphs play a crucial
role in the performances of the data structures \cite{dtree}. 
We distinguish between graphs with small diameters
and graphs with large diameters. Real-world graphs tend to 
have small diameters \cite{konect_diameter}.  
USA Road \cite{road_usa} is a real-world graph with a 
large diameter.
We include representative large real-world graphs from 
various domains (Youtube \cite{mislove09}, 
Stackoverflow \cite{dataset_Stackoverflow}, 
USA Road \cite{road_usa}, 
and Semantic Scholar \cite{ammar}). 
To further enhance the diversity of our datasets, we 
use networkX \cite{networkx} to generate a random graph, a 
power-law graph, a complete graph, a star graph, and a path graph.
Semantic Scholar is used to test the scalability of the data structures. 
Table~\ref{table:datasets} summarizes the
graphs used in our experiments.

\begin{table}[htb!]\centering
  \begin{tabular}{|c|c|c|c|c|}
  \hline
  Name & $|V|$ & $|E|$ & diameter   \\
  \hline
  \textbf{\small Star Graph (SG)} {\scriptsize~\cite{networkx_star}} 
  & $10^7$ & $10^7$ & small \\
  \hline
  \textbf{\small Path Graph (PG)} {\scriptsize~\cite{networkx_path}}
  & $10^7$ & $10^7$ & large\\
  \hline
  \textbf{\small Complete Graph (CG)}{\scriptsize~\cite{networkx_complete}} 
   & 4472 & $10^7$ & small \\
  \hline
  \textbf{\small Random Graph (RG)} {\scriptsize~\cite{networkx_random}} 
   & $10^6$ & $10^7$ & small \\
  \hline
  \textbf{\small Power-law Graph (PL)} {\scriptsize~\cite{networkx_powerlaw}} 
  & $10^6$ & $10^7$ & small\\
  \hline
  \textbf{\small youtube-growth (YT)}{\scriptsize~\cite{mislove09}} & 3.2 $ \times 10^6$ & 1.44$  \times 10^7$ & small \\
  \hline
  \textbf{\small Stackoverflow (ST)}{\scriptsize~\cite{dataset_Stackoverflow}} & 2.6 $ \times 10^6$ & 6.3$  \times 10^7$ & small  \\
  \hline
  \textbf{\small USA Road (USA)}
  {\scriptsize~\cite{road_usa}} &  2.4 $\times 10^7$ &5.7 $\times 10^7$ & large \\
  \hline
  \textbf{\small Semantic Scholar (SC)}
  {\scriptsize~\cite{ammar}} & 6.5 $\times 10^7$ & 8.27$  \times 10^9$ & small \\
  \hline
  \end{tabular}
  \caption{Statistics of datasets.}
  \label{table:datasets}
\end{table}

\subsection{Memory Footprint}
\label{sec:memory_footprint}

The D-tree, LCT, ST and STV are space-efficient data 
structures as they
directly keep track of connected nodes. The other data structures are
deletion-efficient data structures as they optimize the amortized 
costs for finding replacement edges. We measure the memory
footprints of the data structures after all edges have been inserted
and show the experimental results in Figure~\ref{fig:memory_insertions}.

\begin{figure}[htb!]
  \scalebox{0.75} {
     \includegraphics{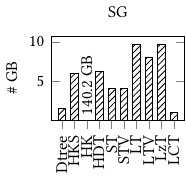}
     \includegraphics{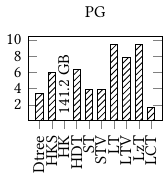}
     \includegraphics{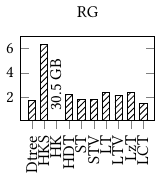}
     \includegraphics{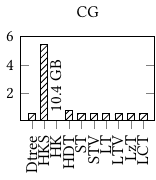}
    }
    \scalebox{0.75} {
     \includegraphics{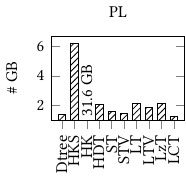}
     \includegraphics{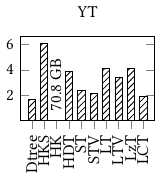}
     \includegraphics{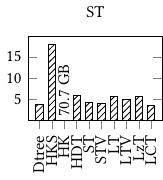}
     \includegraphics{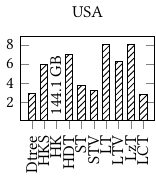}
    }
  \caption{Memory footprints for the data structures 
  of the graphs in Table~\ref{table:datasets}. All data structures 
  except the D-tree run out of memory on SC.  Update performances
    are discussed in Section~\ref{sec:update_performance}.}
  \label{fig:memory_insertions}
\end{figure}

Figure~\ref{fig:memory_insertions} shows that the D-tree has the
lowest memory footprint, followed by ST and STV.  HK has the highest
memory footprint as it maintains ET-trees for $O(\log n)$ levels
and stores non-tree neighbors in random search trees .  
HK uses up to 85x more memory than the
D-tree.  ST and STV have lower memory footprints than LT, LTV and LzT
on all datasets. As shown in Section \ref{sec:lt}, transforming
structural trees to local trees requires additional rank roots and
additional nodes connecting rank roots.  Hence, LT, LTV and LzT
maintain more nodes than ST and STV.  The memory footprints for HKS,
HK and HDT are determined by ET-trees and cumulative spanning trees.
As shown in Section \ref{sec:imp_hks}, additional data structures are
needed to keep track of Euler tours.  Our experiments confirm that
HKS, HK and HDT have higher memory footprints than 
D-tree, LCT and ST. 
No deletion of edges is the best case for HDT as no edges
are pushed down. Even so, HDT has higher memory footprints than the
D-tree, ST and STV on all datasets. When edges are pushed down to
different levels due to the deletion of edges, HDT maintains ET-trees
for cumulative spanning trees at all levels, which can have the same
memory footprints as HK in the worst case. HKS has higher 
memory
footprints than HDT especially on CG and ST as HKS stores non-tree
edges neighbors in random search trees.  On SC, the largest
experiments used in our experiments, only the D-tree finishes the
insertions of all edges while other data structures run out of memory. 
The D-tree uses 155 GB of memory on SC.

Table \ref{tab:space_complexity} shows the space complexity for 
the different data
structures. All data structures need space to store graph 
edges. A D-tree is a spanning tree where the parent and children 
attributes are
tree edges.  A D-tree needs $|E|$ - $|V|$ space to store non-tree
edges and $c *|V|$ space to store nodes. 
Graphs without non-tree edges yields
the lower bound $c *|V|$. Graphs with $|E|$ $=$ $O(|V|^2)$ 
yield the upper bound $|E|$ + $(c - 1)*|V|$.
Parent, children, left and
right attributes in all other data structures are not edges of 
 tree
structures and hence at least $|E|$ space is required to store 
edges. ST, STV, LT, LTV and LzT need exactly 
$|E|$ space to store graph edges. HKS, HK and HDT need more 
than $|E|$ space to store graphs edges since 
edges are stored in randomized search trees.
The lower bound of the space 
complexity for storing nodes in ST, STV, LT, LTV and LzT is
$ c * |V|$ as they need additional space for super 
nodes. HK and HDT
use $c * |V| \times \log |V|$ space for nodes as they maintain
$\log |V|$ ET-trees, one for each level.  HKS and HK need additional
space for storing non-tree neighbors in random search trees. 
Note that constant $c$ differs from method to method.

\begin{table}[htb!]\centering
\scalebox{0.82} {
  \begin{tabular}{|c|c|c|}
    \hline
    Data structures & Lower bound & Upper bound\\ 
    \hline
    D-tree & $c*|V|$ & $|E|$ + $(c - 1)*|V|$\\
    \hline
    LCT & $c*|V|$ & $|E|$ + $(c - 1)*|V|$\\
    \hline
    ST, STV  &  $|E| + c * |V|$ & $|E| + 2 * c * |V|$\\
    \hline
    LT, LTV, LzT & $|E| + c * |V|$ & $|E| + c' * |V|~(c' > c)$ \\
    \hline
    HKS & $|E| + c * |V|$ & $k*|E| + c * |V|$ , $k > 1$\\
    \hline
    HDT & $|E| + c * |V|$ & $k*|E| + c * |V| * \log |V|$, $k \geq 1$ \\
    \hline
    HK & $|E| + c * |V|$ & $k*|E| + c * |V| * \log |V|$, $k > 1$\\
    \hline
  \end{tabular}
}  
\caption{Space complexity for the data structures. Note that
    $|V|$ $\leq$ $|E|$ $\leq$ $|V|*(|V|-1)/2$.}
  \label{tab:space_complexity}
\end{table}

\subsection{Performance of Update Operations} 
\label{sec:update_performance}

We generate workloads that consist of sequences of
  interleaved insertions and deletions.  Since no general workloads
  exist we created them in a way that makes it possible to vary the
  growth rate of graphs and makes insertions and deletions independent
  of each other (workloads that delete all inserted edges are not
  general since many graphs grow \cite{leskovec07tkdd,
    leskovec2005graphs}; workloads that delete edges a certain time
  after they have been inserted ~\cite{dtree, song24,
  zhang2024incremental} render insertions and deletions
  dependent on each other).  We use $u_r$ =
  $\frac{\# insertions}{\# deletions}$, i.e., one deletion occurs per
  $u_r$ insertions, to quantify the growth of the graph.  We start out
  with the empty graph and insert the edges of the graphs in
  Table~\ref{table:datasets}.  When deleting an edge, we randomly
  select an edge from the graph.  Small $u_r$ values mean that edges
are deleted frequently and therefore generate graphs with small
connected components.  Processing small components is fast and can be
done with brute-force approaches and is not the focus of our paper.
Large $u_r$ values mean that edges are rarely deleted.  For our
experiments, we set $u_r$ to 1000, 100, 20, 5 to cover representative
scenarios where graphs are growing quickly ($u_r$ = 1000) and slowly
($u_r$ = 5).

When a workload is run we measure and sum up the runtimes
  of, respectively, insertions and deletions.  These times are used to
  calculate the average runtime of insertions and deletions for each
  data structure.  We first compare data structures in the same
category (for details see Section~\ref{sec:maintaining_ds}) 
and determine the best data structure of each category for
  further comparisons.  We compare the best performing structures,
and give our guidance for using the data structures.  When evaluating
the performance of update operations, 
we do not include SC since only the
D-tree can finish the workloads of updates within a few hours. The
other approaches can not finish the workloads in several days and end
up running out of memory.

\subsubsection{D-trees and LCTs}
\label{sec:update_dtree_lct}

Figure~\ref{fig:update_dtrees} shows the average run time for 
insertions and deletions for D-trees. Inserting edges on 
PG and deleting edges on USA is up to several orders of 
magnitude slower than other datasets. Both PG and USA 
have large diameters and the D-tree degenerates. 
The reason is that 
vertices of graphs with large diameters are connected through 
long paths and consequently the spanning trees for these graphs 
include long paths. Updating D-trees demands the 
traversals of nodes on long paths, which is inefficient.

\begin{figure}[htb!]
  \begin{subfigure}[htb]{1.0\columnwidth}\centering
    \scalebox{0.85} {
     \includegraphics{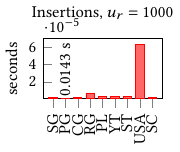}
     \includegraphics{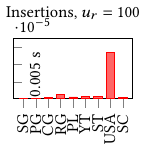}
     \includegraphics{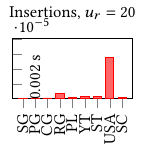}
     \includegraphics{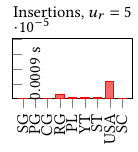}
    }
  \end{subfigure}
  \begin{subfigure}[htb]{1.0\columnwidth}\centering
    \scalebox{0.80} {
     \includegraphics{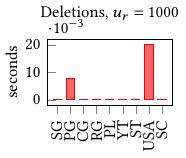}
     \includegraphics{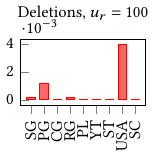}
     \includegraphics{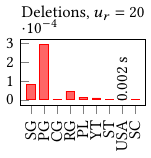}
     \includegraphics{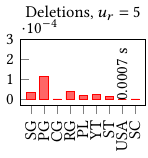}
    }
  \end{subfigure}
  \caption{Average runtime of update operations for D-trees. }
  \label{fig:update_dtrees}
\end{figure}

\begin{figure}[htb!] \centering
  \begin{subfigure}[b]{1.0\columnwidth}
    \includegraphics{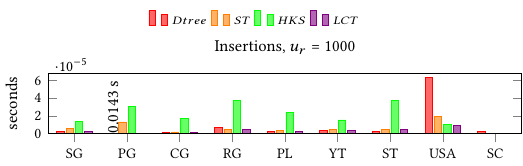}
  \end{subfigure}
  \begin{subfigure}[b]{1.0\columnwidth}
    \includegraphics{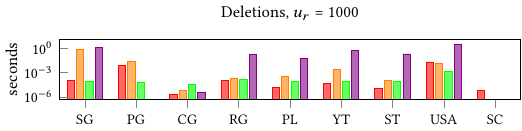}
  \end{subfigure}

  \caption{Average runtime of update operations for the
      D-tree, LCT, ST and HKS with $u_r$ = 1000. LCT runs out of time
      on PG.  LCT, HKS, and ST run out of memory on SC.}
  \label{fig:dtree_st_hks_lct}
\end{figure}

Figure~\ref{fig:dtree_st_hks_lct} shows that LCT is dominated
by D-tree and ST in terms of update performances. 
Deleting edges in LCT is up to $10^4$ times
slower than other data structures.  Query performances for LCT are 
dominated by ST.

\subsubsection{HKS, HK and HDT}

When evaluating the performance of update operations for HK, we find
that updating HK is up to several orders of magnitude slower than HKS
and HDT as HK constantly rebuilds the ET-trees across different
levels.  Hence, we do not show the update performances for HK.  Update
performances for HKS and HDT are shown in Figure~\ref{fig:HKS_HK_HDT}
and we find that HKS and HDT have similar update performances. HKS and
HDT are comparable in inserting edges.  Performances of deletions for
HKS and HDT are similar, which means that HDT optimizations do not pay
off.  HDT is up to several orders of magnitude slower than HKS on PG
where all edges are tree edges.  For example, HDT is 1335x slower than
HKS on PG with $u_r$ = 1000. The deletion of any edge in PG splits one
path and HDT pushes down all nodes in the smaller path. Compared to
HDT, HKS does not push down edges, which can save lot of time.  HKS
outperforms HDT when edges are frequently deleted, $i.e.,$ $u_r$ = 5.
Since HKS is simpler than HDT, we choose HKS over HDT.

\begin{figure}[htb!] \centering
  \begin{subfigure}[b]{0.55\columnwidth}
    \includegraphics{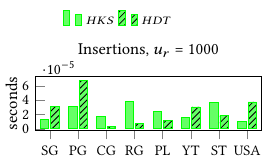}
    \includegraphics{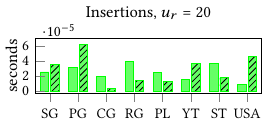}
  \end{subfigure}
  \begin{subfigure}[b]{0.42\columnwidth}
    \includegraphics{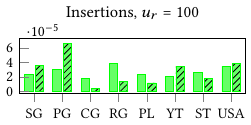}
    \includegraphics{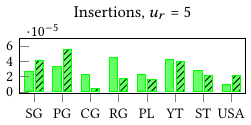}
  \end{subfigure}
  \begin{subfigure}[b]{0.55\columnwidth}
    \includegraphics{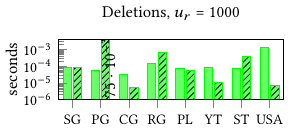}
    \includegraphics{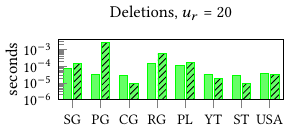}
  \end{subfigure}
  \begin{subfigure}[b]{0.42\columnwidth}
    \includegraphics{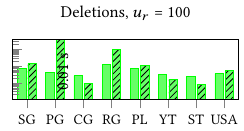}
    \includegraphics{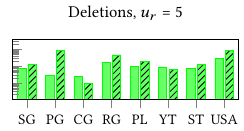}
  \end{subfigure}
  \caption{Average runtime of update operations for HKS and HDT. 
  HKS and HDT run out of memory on SC. }
  \label{fig:HKS_HK_HDT}
\end{figure}

\subsubsection{ST, STV, LT, LTV and LzT} 

Figure \ref{fig:res_ST_LT} shows the update performances for ST, STV,
LT, LTV and LzT when $u_r$ = 1000 and 100.  We do not show evaluations
for $u_r$ = 20 or 5 as LT, LTV and LzT run out of time on these
workloads for some datasets.  Our experiments show that ST outperforms
LT and LzT in both insertions and deletions. It shows that the LT
and LzT optimizations do not improve but slow down the
update performances. The reason is that removing nodes from rank 
trees, adding nodes into rank trees and merging rank trees are more
complicated than ST operations.  For ST and LT, differentiating between tree
edges and non-tree edges is not efficient in practice.  ST is
comparable to STV in inserting edges over all datasets except the 
USA dataset. ST is up to 1300x faster than STV for deleting
edges. LT outperforms LTV in inserting and deleting edges. Since ST
outperforms LT, LzT, STV and LTV, ST is the best in this category. 
We find that deletion operations in ST, STV, LT, LTV
and LzT can take extremely long time. For
example, deleting a particular edge in LT on 
YT takes $10^4$ seconds. The
deletion outliers are caused by pushing down a large number of nodes. 
ST is considered the best of this category.

\begin{figure}
  \includegraphics{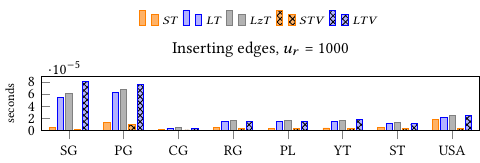}
  \includegraphics{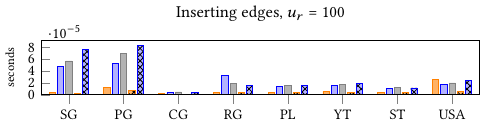}
  \includegraphics{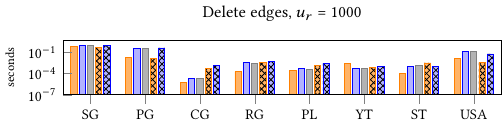}
  \includegraphics{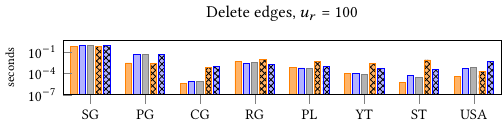}
  \caption{Average runtime of update operations 
  for ST, STV, LT, LTV and LzT. ST, LT, LzT, STV and LTV 
  run out of memory on SC. }
  \label{fig:res_ST_LT}
\end{figure}

\subsubsection{D-tree, ST and HKS}
\label{sec:dtree_st_hks}
We compare D-tree, HKS and ST, the best of each
category.
Figure \ref{fig:dtree_st_hks} shows the update performances
for D-tree, ST and HKS, 
with $u_r$ = 1000 and $u_r$ = 100. Inserting edges in ST is the fastest.
Inserting edges in D-trees is comparable to ST 
except on PG and USA, which are graphs with large diameters. 
Inserting edges in HKS is the slowest and 
up to one order of magnitude slower than ST.
The D-tree outperforms ST and HKS for deleting edges 
for all datasets except USA and PG. 
In general, the D-tree outperforms ST and HKS 
on graphs with small diameters. 
For PG and USA, ST outperforms HKS in insertions 
while HKS outperforms ST in deletions.

\begin{figure}[htb!] \centering
  \begin{subfigure}[b]{0.55\columnwidth}
    \includegraphics{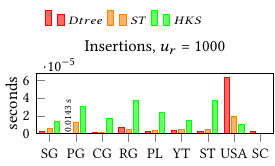}
  \end{subfigure}
  \begin{subfigure}[b]{0.42\columnwidth}
    \includegraphics{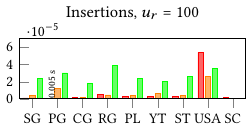}
  \end{subfigure}
  \begin{subfigure}[b]{0.55\columnwidth}
    \includegraphics{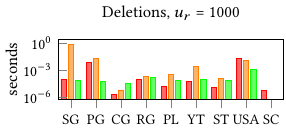}
  \end{subfigure}
  \begin{subfigure}[b]{0.42\columnwidth}
    \includegraphics{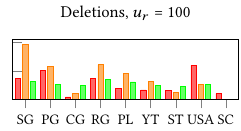}
  \end{subfigure}

  \caption{Average runtime of update operations 
  for the D-tree, ST and HKS. 
  HKS and ST run out of memory on SC.}
  \label{fig:dtree_st_hks}
\end{figure}

\subsection{Performance for Connectivity Queries}
\label{sec:query_performance}

We evaluate query performances during the life span of 
updates and we choose \texttt{test\_num} = 100 \emph{testing points} 
that are uniformly distributed among the updates. When running
queries, we suspend the updates. The updates are resumed at 
each testing point after the queries have been executed.
At each testing point we run 50 million 
connectivity queries for uniformly
distributed pairs of vertices and report the average runtime since
running connectivity queries for all pairs of vertices at 
each testing point is impractical.
Let $N_u$ be the total number of updates, $i.e.,$
the insertion and deletions of edges. 
This means there is one testing point per
$\sfrac{N_u}{test\_num}$ updates. We evaluate query 
performances for the D-tree, ST and HKS. We only  
show query performances with
$u_r$ = 1000 in Figure~\ref{fig:query_performance}
as experiments with $u_r$ = 100, 20 and 5 show very
similar results.  Our experiments show that ST is the 
best for query performances as its tree height 
is guaranteed to be at most $\log n$. ST outperforms all other 
data structures in answering queries on all datasets.
Query performance for the D-tree on PG is not shown 
because at each testing point it takes several hours for 
the D-tree to run queries, which is at least two 
orders of magnitude slower than 
ST. On PG and USA, which are worst cases for 
the D-tree, the query 
performances for the D-tree degenerate. 

\begin{figure}[htb!]\centering
  \includegraphics{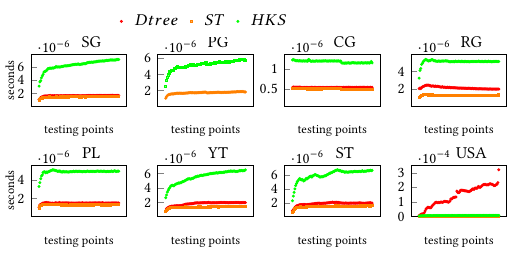}
  \caption{Average runtime of queries for the D-tree,
    HKS, and ST with $u_r$ = 1000.}
  \label{fig:query_performance}
\end{figure}

\subsection{Lessons Learned and Recommendations}
\label{sec:disc-recomm}

The first lesson is that lazy local trees with the lowest
  amortized costs are the slowest in terms of actual runtimes.  Lazy
  local trees address the worst case performance when an edge is deleted
  and a replacement edge must be searched, but the overhead for lazy
  local trees is substantial and the overall performance suffers.  In
  general, solutions that optimize amortized costs, either implemented
  with Euler-tour trees or height bounded trees, come at a hefty cost.
  For the workloads they also make the assumption 
  that all edges are deleted, which is not
  true in general.  Refined theoretical studies are needed to
  determine the amortized costs for a widely agreed-upon benchmark and
  workload for fully dynamic graphs.  Since no such benchmarks exist,
  the first critical step is to design and establish a standardized
  benchmark and workloads 
  that can be used to empirically evaluate solutions.

  The second lesson is that a good balance between space-efficient and
  deletion-efficient data structures has been elusive so far.  Our
  experiments show that the overhead of all deletion-efficient data
  structures is too high. We find that for workloads with a higher
  percentage of edge deletions (edge deletions are the most expensive
  operations) the runtime of the workloads on graphs with large
  diameters improves for unbalanced spanning trees. 
  The reason is that edge deletions
  break up long paths into shorter paths. This is effective for even
  a small number of deletions. 
  A lightweight modification of space-efficient data
  structures to curb the diameter seems the most promising direction
  for future research.  Our recommendation is to enhance unbalanced
  spanning trees 
   to avoid deterioration in the worst case. 

A third lesson is that constructing and maintaining a partitioned
structure, i.e., spanning trees where edges are placed on different
levels, is expensive.  Researchers who first implemented and evaluated
the original Euler-tour tree solutions also observed this and have
suggested the simplified Euler-tour tree solution without
levels.  While a partitioning of spanning trees with a push down of
edges to a higher level is theoretically attractive more work is
needed to possibly make this a practically attractive solution.

The fourth lesson are the costs for, respectively,
  computing and updating the attribute values of spanning trees (size,
  weight, bitmap, rank, etc).  The attributes are critical elements
  for pruning the search space and they can either be computed on
  demand or materialized.  Computing on demand or updating
  materialized attribute values in dynamic settings incur very
  different costs.  The design choices for supporting attributes must
  be based on theoretical and empirical results.  Implementing an
  algorithm without carefully justified design decisions provides
  partial (and possibly misleading) insights \cite{Hanauer22}.

\section{Related Work}

The existing data structures for answering connectivity queries on
fully dynamic graph have been described in detail in the previous
sections.  All data structures handle real-time deletions and
insertions of edges and hence answer connectivity queries in
real-time.  Recent work \cite{zhang2024incremental, song24} proposes
indices for connectivity queries over a batch of snapshots of the
graph and show that the D-tree also outperforms HK and
HDT in such settings. 
An alternative to spanning trees are labeling schemes
  for reachability queries \cite{bonifati2022querying, su_tkde}.
  Usually such approaches focus on directed graphs, and their
  performance depends on the label size. Update performances for
  labeling schemes can be quadratic while for approaches based on
  spanning trees, update performances are usually polylogarithmic or
  linear.  Labeling schemes do not perform well for fully dynamic
  graphs since the deletions of edges triggers an exhaustive search
  for updating the labels~\cite{dtree}.

Existing experimental studies are outdated, and were done on small graphs
with limited workloads, and do not cover D-tree, structural trees,
local trees and lazy local trees.  Alberts at
al. \cite{David_HKvariant} evaluated HK, HKS and a simplified
sparsification technique \cite{eppstein1992sparsification} using
random and non-random workloads over selected synthetic graphs.  The
number of updates is from 50 to 5000.  Their experiments show that HK
is faster than the sparsification for a larger number of updates.
Iyer at al. \cite{Iyer2002experimental} compared HK with HDT on random
graphs, semi-random graphs and cliques with at most 70,000
vertices. Random and non-random workloads over the synthetic graphs
are used.  The number of updates is from 10 to 1,000,000.  Their
experiments show that HDT and HK are comparable for a large number of
updates.

\section{Conclusion}
\label{sec:conclusion}

In this paper we comprehensively evaluate all the major data
structures for the dynamic connectivity problem by comparing their
memory footprints, update performances and query performances.  We
find that none of the data structures is robust in practice.  The 
D-tree and structural tree degenerate in worst case scenarios
 (line graph or other graphs with large diameters)
 while all balanced data structures are too 
 expensive to maintain.
Data structures based on ET-tree, especially HK, have high memory
footprints.  A robust data structure should have a bounded tree height
  (a cheap loose bound is better than an 
 expensive tight bound) 
 and the maintenance of such data
structure should not be expensive. In the future,
we plan to leverage our findings to advance 
existing data structures to a practical and robust 
solution and integrate
it into an existing DBMS such as K\`uzu \cite{kuzucidr}. 
As part of future research it is also interesting to 
investigate  
solutions for fully dynamic undirected graphs that are applicable
to directed graphs and to work on 
scalable labelings for the reachability 
problem that perform well for fully dynamic undirected graphs.

\section{Acknowledgements}
The authors would like to thank Jürgen Bernard for his valuable 
suggestions for graphically illustrating the relationships between 
the data structures. Qing Chen was partially supported by University of Zürich CanDoc grant 
Nr. FK-23-021. 


\bibliographystyle{ACM-Reference-Format}
\interlinepenalty=10000
\bibliography{paperbib}

\section{Appendix}

\subsection{Complementary materials for the local tree}
\label{sec:app_lt}

Algorithm \ref{alg:pair} shows the procedure of creating a local rank
root by pairing up two nodes with the same rank.
\begin{algorithm2e}[htb!]
  \small
  \caption{Pair($x$, $y$)}
  \label{alg:pair}
  \SetKwInOut{Input}{input}
  \SetKwInOut{Output}{output}
  \SetKwRepeat{Do}{do}{while}
  \Input{$x$ and $y$ are two nodes with same rank}
  \Output{$par$: a rank root that is the parent of $x$ and $y$}
  initialize a new node $par$\\
  $par.left$ = $x$; $par.right$ = $y$\\
  $x.parent$ = $par$; $y.parent$ = $par$\\
  $par.nl$ = $x.nl$ + $y.nl$\\
  $rank(par)$ = $rank(x)$ + 1\\
  $par.bitmap_t$ = $x.bitmap_t$ $|$ $y.bitmap_t$ (bitwise OR)\\
  $par.bitmap_{nt}$ = $x.bitmap_{nt}$ $|$ $y.bitmap_{nt}$ (bitwise OR)\\
  \Return $par$\\
\end{algorithm2e}

Next, we show our proofs for the local tree properties.

\begin{lemma}
\label{lemma:depth_localrankroot}
In a rank tree with the root node $r$, the depth of a node $x$ 
$depth(x)$ $=$ $rank(r)$ - $rank(x)$.
\end{lemma}

\begin{proof}
$W.l.o.g$, consider the non-root node $x$ and its 
parent $par$. As shown by line 5 in 
Algorithm~\ref{alg:pair}, $rank(par)$ $-$ $rank(x)$ $=$ 1 which
is equal to $depth(x)$ $-$ $depth(par)$. We use $anc(x, dx)$ to
denote the ancestor of $x$ with the distance $dx$. 
Let $d_x$ be the depth of $x$ in the local rank tree, hence
{\tiny 
\begin{align*}
  depth(x) - depth(anc(x, 1)) &= rank(anc(x, 1)) - rank(x)\\
  depth(anc(x, 1)) - depth(anc(x, 2)) &= rank(anc(x, 2)) - rank(anc(x, 1))\\
  & \dots\\
  depth(anc(x, d_x - 2)) - depth(anc(x, d_x - 1)) &= rank(anc(x, d_x - 1)) - rank(anc(x, d_x - 2))\\
  depth(anc(x, d_x - 1)) - depth(anc(x, d_x)) &= rank(anc(x, d_x)) - rank(anc(x, d_x - 1))\\
\end{align*}
}

Summing up, 
\[ depth(x) - depth(anc(x, d_x)) = rank(anc(x, d_x)) - rank(x)\]
Since $anc(x, d_x)$ = $r$ and $depth(anc(x, d_x))$ $=$ $depth(r)$ $=$ $0$, 
the depth of $x$ in the local rank tree 
\[depth(x) = rank(r) - rank(x)\]
\end{proof}

\begin{lemma}[\cite{thorup2000near}]
\label{lemma:ranktree_height}
The height of a rank tree is at most $\log n$.
\end{lemma}
\begin{proof}
$W.l.o.g$, let $y$ be one leaf node in a rank tree with the root node $r$. 
It follows that $r.nl$ $\leq$ $n$ and
$y.nl$ $\geq$ 1. With Lemma~\ref{lemma:depth_localrankroot}, 
\begin{align*}
  depth(y) & \leq \log n - \log 1 \\
           & \leq \log n
\end{align*}

\end{proof}

\begin{lemma}[\cite{thorup2000near}]
\label{lemma:localrankroot}
There are at most $\log n$ rank roots with unique
ranks for a local tree with $n$ tree nodes.
\end{lemma}

\begin{proof}
Let $S$ be the list of rank roots with unique ranks 
and $|S|$ $=$ $X$.  By the definition
of rank, $rank(S[i])$ = $\lfloor$ $\log_2 S[i].nl$ $\rfloor$. 
It follows that $S[i].nl$ $\geq$  $2^{rank(S[i])}$.
The fact is  $n$ $=$ $\sum\limits_{i=0}^{X - 1} S[i].nl$.
We prove by contradiction. 
Assume the opposite that $X$ $>$ $\log n$.
\begin{itemize}
  \item If $X$ $=$ $\log n$ $+$ 1, then $rank(S[i]) = i$ for $i$ $<$ $X$.
  \begin{align*}
  \sum\limits_{i=0}^{X - 1} S[i].nl &\geq \sum\limits_{i=0}^{\log n} 2^{rank(S[i])} \\
  &\geq \sum\limits_{i=0}^{\log n} 2^i\\
  &\geq 2 * n > n
  \end{align*}
  \item If $X$ $>$ $\log n$ $+$ 1, $rank(S[X - 1])$ $>$ $\log n$ as
  every local rank root has a unique rank. 
  By definition, $S[X - 1].nx$
   $\geq 2^{rank(S[X - 1])}$ $>$ $n$. 
\end{itemize}
To conclude, $X$ $\leq$ $\log n$. Hence $|S|$ $\leq$ $\log n$.
\end{proof}

Algorithm \ref{alg:constructlocaltree} shows the construction of
a local tree over a sorted list of local rank roots with unique ranks.
\begin{algorithm2e}[htb!]
  \small
  \caption{Construct($S$, $a$)}
  \label{alg:constructlocaltree}
  \SetKwInOut{Input}{input}
  \SetKwInOut{Output}{output}
  \SetKwRepeat{Do}{do}{while}
  \Input{$S$: a sorted list of rank roots with unique ranks}
  \Output{a: the root node of a local tree}
  $X$ = $|S|$\\
  Initialize a node $a$ \\
  \If{$X$ = 1}{
    \Return $a$ with $S[0]$ as the left child\\
  } \Else {
    $lc$ = $S[0]$, $rc$ = $S[1]$\\
    \If{$X$ == 2}{
      $cur$ $=$ $a$ \\
    }\Else{
      Initialize a new node $cur$\\ 
    }
    make $lc$ and $rc$ as the left and right child of $cur$, respectively\\
    $i$ = 2\\
    \While{$i$ $<$ $X$} {
      \If{$i$ == $X - 1$} {
        $ne$ = $a$ (the root node)
      } \Else {
        Initialize a new node $ne$ (a connecting node)
      }
      $ne.left$ = $cur$\\
      $ne.right$ = $S[i]$\\
      $i$ = $i$ $+$ 1\\
      $cur$ = $ne$\\
    }
    \Return $cur$\\
  } 
\end{algorithm2e}

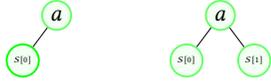
\begin{figure}[htb!]\centering
  \scalebox{0.6}{
    \begin{tikzpicture}[roundnode/.style={circle, draw=green!60, fill=green!5, very thick, minimum size=6mm}]
      \node[roundnode, draw=green] (l0)  at (-1.5, 0.0)  {{\tiny$S[0]$}};
      \node[roundnode] (p1)  at (-0.75, 1.0)  {$a$};
      \draw[-] (p1) to (l0);
    \end{tikzpicture}
    \hspace{2cm}
    \begin{tikzpicture}[roundnode/.style={circle, draw=green!60, fill=green!5, very thick, minimum size=6mm}]
      \node[roundnode] (l0)  at (-1.5, 0.0)  {{\tiny$S[0]$}};
      \node[roundnode] (l1)  at (0.0, 0.0)  {{\tiny$S[1]$}};
      \node[roundnode] (p1)  at (-0.75, 1.0)  {$a$};
      \draw[-] (p1) to (l0);
      \draw[-] (p1) to (l1);
    \end{tikzpicture}
  }

  \caption{A local tree constructed by Algorithm~\ref{alg:constructlocaltree} when $X$ $\leq$ $2$.}
\end{figure} 
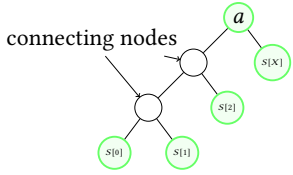
\begin{figure}[htb!]\centering
  \scalebox{0.6}{
    \begin{tikzpicture}[roundnode/.style={circle, draw=green!60, fill=green!5, very thick, minimum size=6mm}]
      \node[roundnode] (l0)  at (-1.5, 0.0)  {{\tiny$S[0]$}};
      \node[roundnode] (l1)  at (0.0, 0.0)  {{\tiny$S[1]$}};
      \node[circle, draw, minimum size=6mm] (p1)  at (-0.75, 1.0)  {};
      \draw[-] (p1) to (l0);
      \draw[-] (p1) to (l1);
      
      \node[roundnode] (l2)  at (1.0, 1.0)  {{\tiny$S[2]$}};
      \node[circle, draw, minimum size=6mm] (p2)  at (0.25, 2.0)  {};
      \draw[-] (p2) to (p1);
      \draw[-] (p2) to (l2);
      
      \node[roundnode] (r)  at (1.25, 3.0)  {$a$};
      \draw[-] (r) to (p2);
      \node[roundnode] (lx)  at (2.0, 2.0)  {{\tiny$S[X]$}};
      \draw[-] (r) to (lx);
      \node(l) at (-2, 2.5) {connecting nodes};
      \draw [->] (l) edge (p1);
      \draw [->] (l) edge (p2);
    \end{tikzpicture}
  }
  \caption{A local tree constructed by Algorithm~\ref{alg:constructlocaltree} for $X$ > 2.}
\end{figure}

\begin{lemma}
  \label{lemma:dist_lrt}
  In the local tree with the root node $r$, the depth for a local rank root
  $rt$ $depth(rt)$ $\leq$ 1 $+$ $rank(r)$ $-$ $rank(rt)$.
\end{lemma}

\begin{proof}
Let $S$ be the list of rank roots with unique ranks 
and $|S|$ $=$ $X$. \\
\noindent \textbf{Case 1: $X$ $\leq$ 2} \\
Depths for $S[0]$ and $S[1]$:
\begin{itemize}
  \item $depth(S[0])$ $=$ $1$ $\leq$ $1$ $+$ $rank(r)$ $-$ $rank(S[0])$ 
  \item $depth(S[1])$ $=$ $1$ $\leq$ $1$ $+$ $rank(r)$ $-$ $rank(S[1])$
\end{itemize}
Hence, $depth(S[i])$ $\leq$ $1$ $+$ $rank(r)$ $-$ $rank(S[i])$.

\noindent \textbf{Case 2: $X$ $>$ 2}\\
A local tree is constructed such that $depth(S[i])$ 
$-$ $depth(S[i + 1])$ $=$ 1 for $i$ $>$ 0. Since rank roots in $S$ have 
unique ranks, $rank(S[i + 1])$ $>$ $rank(S[i])$, 
more precisely $rank(S[i + 1])$ 
$-$ $rank(S[i])$ $\geq$ 1. It follows that $depth(S[i])$ - 
$depth(S[i + 1])$ $\leq$ $rank(S[i + 1])$ - $rank(S[i])$.

In general, 
\begin{align*}
  depth(S[i]) - depth(S[i + 1])&\leq rank(S[i + 1]) - rank(S[i])\\
  depth(S[i + 1]) - depth(S[i + 2]) &\leq rank(S[i + 2]) - rank(S[i + 1])\\
  & \dots\\
  depth(S[X-1]) - depth(S[X])&\leq rank(S[X]) - rank(S[X-1])\\
  depth(S[X]) - depth(r) &\leq rank(r) - rank(S[X]) + 1\\
\end{align*}
To sum up, $depth(S[i]) - depth(r) \leq rank(r) - rank(S[i]) + 1$. 
Since $depth(r)$ $=$ 0, $depth(S[i]) \leq rank(r) - rank((S[i]) + 1$.

To conclude, $depth(S[i]) \leq rank(r) - rank(S[i]) + 1$.
\end{proof}

\begin{theorem}[~\cite{thorup2000near}]
\label{thm:localtree_property}
In a local tree rooted with the root node $r$ 
the depth of a node $x$ is $O(1 + \log(r.nl/x.nl))$.
\end{theorem}

\begin{proof}~\\
\begin{itemize}
  \item If node $x$ is a rank root, Lemma~\ref{lemma:depth_localrankroot} 
  shows that $depth(x)$ $\leq$ $rank(r)$ - $rank(x)$ $+$ 1.
  \item Otherwise, let $lt$ be the root node of the rank tree that 
  contains $x$. The $depth(x)$ $=$ $dist(x, lt)$ $+$ $dist(lt, r)$, where
    $dist(x, lt)$ is the $x$'s depth in the local rank tree rooted at $lt$
    and $dist(lt, r)$ is the depth of $lt$ in the local tree.
Lemma~\ref{lemma:dist_lrt} shows that $dist(lt, r)$ $\leq$ $rank(r)$ - $lt$.rank + 1
and Lemma~\ref{lemma:depth_localrankroot} shows that $dist(x, lt)$ $=$ $rank(lt)$ - $rank(x)$.
Hence,
\begin{align*}
depth(x) &= dist(x, lt) + dist(lt, r)\\
          &= rank(lt) - rank(x) + dist(lt, r)\\
          &\leq rank(lt) - rank(x) + rank(r) - rank(lt) + 1\\
          &\leq rank(r) - rank(x) + 1 \\
\end{align*}
\end{itemize}
In general,
\begin{align*}
depth(x) &\leq rank(r) - rank(x) + 1 \\
          &\leq \lfloor \log_2 r.nl \rfloor - \lfloor \log_2 x.nl \rfloor + 1\\
          &= O(\log r.nl - \log_2 x.nl + 1)\\
          &= O(\log  \frac{r.nl}{x.nl} + 1)\\
\end{align*}
\end{proof} 

\begin{lemma}
\label{lemma:height_localtree}
The height of a local tree is $O(\log n)$.
\end{lemma}

\begin{proof}
Let $x$ be the node with maximum depth in the local tree 
rooted at $r$. The height of the local tree is equal to 
the $depth(x)$. Theorem~\ref{thm:localtree_property} shows 
that $depth(x)$ = $O(\log \frac{r.nl}{x.nl}) + 1$.
Node $x$ is a leaf node to have the maximum depth, 
hence $n(x)$ $=$ $1$. It follows that $depth(x)$ $=$ 
$O(\log r.nl) + 1$ $=$ $O(\log r.nl)$ $=$ $O(\log n)$.
\end{proof}

\begin{figure}[htb!]
  \scalebox{0.7}{
    \begin{subfigure}[b]{0.45\columnwidth}\centering
    \begin{tikzpicture}[roundnode/.style={circle, draw=green!60, fill=green!5, very thick, minimum size=6mm}]
      \node[roundnode] (l0)  at (-1.5, 0.0)  {};
      \node[roundnode] (l1)  at (0.0, 0.0)  {$x$};
      \node[circle, draw, minimum size=6mm] (p1)  at (-0.75, 1.0)  {};
      \draw[-] (p1) to (l0);
      \draw[-] (p1) to (l1);
      
      \node[roundnode] (l2)  at (1.0, 1.0)  {};
      \node[circle, draw, minimum size=6mm] (p2)  at (0.25, 2.0)  {};
      \draw[-] (p2) to (p1);
      \draw[-] (p2) to (l2);
      
      \node[roundnode] (r)  at (1.25, 3.0)  {$a$};
      \draw[-] (r) to (p2);
      \node[roundnode] (lx)  at (2.0, 2.0)  {};
      \draw[-] (r) to (lx);
    \end{tikzpicture}
    \caption{Node $x$ is a local rank root}
    \end{subfigure}
  }
  \hfill
  \scalebox{0.7}{
  \begin{subfigure}[b]{0.5\columnwidth}\centering
    \begin{tikzpicture}[roundnode/.style={circle, draw=green!60, fill=green!5, very thick, minimum size=6mm}]
      \node[roundnode] (l0)  at (-1.5, 0.0)  {};
      \node[roundnode] (l1)  at (0.0, 0.0)  {};
      \node[circle, draw, minimum size=6mm] (p1)  at (-0.75, 1.0)  {};
      \draw[-] (p1) to (l0);
      \draw[-] (p1) to (l1);
      
      \node[roundnode] (l2)  at (1.2, 1.2)  {$lt$};
      \node[circle, draw, minimum size=6mm] (p2)  at (0.25, 2.0)  {};
      \draw[-] (p2) to (p1);
      \draw[-] (p2) to (l2);

      \node[roundnode] (x)  at (1.75, 0.25)  {$x$};
      \node[roundnode] (y)  at (0.75, 0.25)  {};
      \draw[-] (x) to (l2);
      \draw[-] (y) to (l2);
      
      \node[roundnode] (r)  at (1.25, 3.0)  {$a$};
      \draw[-] (r) to (p2);
      \node[roundnode] (lx)  at (2.0, 2.0)  {};
      \draw[-] (r) to (lx);
    \end{tikzpicture}
    \caption{Node $x$ is not a local rank root}
  \end{subfigure}
  }
\end{figure}

\end{document}